\documentclass[final]{l4dc2021}

\jmlrvolume{}
\jmlryear{}
\jmlrproceedings{PMLR}{}
\editors{}
\jmlrworkshop{Accepted to 3rd Annual Conference on Learning for Dynamics and Control --- Long Version}
\makeatletter
\renewcommand{\@titlefoot}{}
\makeatother

\usepackage{times}

\usepackage[english]{babel}
\usepackage[utf8]{inputenc}
\usepackage{enumitem}

\usepackage{amsfonts}       

\usepackage[section]{placeins}
\usepackage{algorithm}
\usepackage[noend]{algpseudocode}

\usepackage[utf8]{inputenc} 
\usepackage[T1]{fontenc}    
\usepackage{booktabs}       
\usepackage{nicefrac}       
\usepackage{microtype}      

\usepackage{xcolor}
\usepackage{graphicx}
\graphicspath{ {./images/} }
\usepackage{tikz}
\usepackage{wrapfig}

\usepackage[font=small,labelfont=bf]{caption}

\usepackage{siunitx}

\usepackage{cleveref}		

\newtheorem{thm}{Theorem}
\newtheorem{lem}{Lemma}

\newtheorem{defn}{Definition}
\newtheorem{assumpn}{Assumption}\Crefname{assumpn}{Assumption}{Assumptions}
\crefname{cor}{Corollary}{Corollaries}
\newtheorem{prb}{Problem}\crefname{prb}{Problem}{Problems}

\allowdisplaybreaks

\title[Feedback from Pixels]{\LARGE \bf Feedback from Pixels:\\
       \textit{Output Regulation via Learning-Based Scene View Synthesis}}
\author{%
 \Name{Murad Abu-Khalaf} \Email{murad@csail.mit.edu}\\
 \addr MIT Computer Science and Artificial Intelligence Laboratory\\
 \addr Cambridge, MA 02139, USA
 \AND
 \Name{Sertac Karaman} \Email{sertac@mit.edu}\\
 \addr MIT Laboratory for Information and Decision Systems\\
 \addr Cambridge, MA 02139, USA
 \AND
 \Name{Daniela Rus} \Email{rus@csail.mit.edu}\\
 \addr MIT Computer Science and Artificial Intelligence Laboratory\\
 \addr Cambridge, MA 02139, USA
}

\begin{document}

\maketitle

\begin{abstract}%
 We propose a novel controller synthesis involving feedback from pixels, whereby the measurement is a high dimensional signal representing a pixelated image with Red-Green-Blue (RGB) values. The approach neither requires feature extraction, nor object detection, nor visual correspondence. The control policy does not involve the estimation of states or similar latent representations. Instead, tracking is achieved directly in image space, with a model of the reference signal embedded as required by the internal model principle. The reference signal is generated by a neural network with learning-based scene view synthesis capabilities. Our approach does not require an end-to-end learning of a pixel-to-action control policy. The approach is applied to a motion control problem, namely the longitudinal dynamics of a car-following problem. We show how this approach lend itself to a tractable stability analysis with associated bounds critical to establishing trustworthiness and interpretability of the closed-loop dynamics.
\end{abstract}


\section{Introduction} \label{Section:Intro}
Our aim is to investigate the integration of visual signals into feedback loops for the purpose of controller synthesis and analysis, and without requiring a perception module in the loop. We treat the camera as a high-dimensional sensor and propose a principled approach grounded in mathematical control theory to investigate stability and associated theoretical limitations of the closed-loop performance. 

In this paper, we consider output regulation class of problems where the output measurement includes a pixelated image. We feel the contribution of this paper is as follows:

\begin{itemize}[noitemsep]
  \item We treat each RGB pixel as a measurement and do not attempt to grayscale or threshold the image and can handle an arbitrary image size or resolution.
  \item Compared to visual servoing approaches, our work does not involve hand-crafted geometrical feature extractions, correspondence or matching, pose estimation, or an interaction matrix.
  \item Unlike most existing approaches, we integrate vision into reactive low-level control without a need for a perception module, end-to-end imitation learning, the estimation of states or similar latent representations.
  \item Our approach works for moving targets and non-stationary environments.
  \item Embedded in our controller is an internal model of the tracked visual reference. This is achieved by incorporating a view synthesizer in the loop at inference or execution time.
  \item We show a systematic way to synthesize static output feedback controllers, such as a proportional controller, via necessary and sufficient conditions in the literature.
  \item Our approach does not require discretizing the action space or the state space, and works in continuous-time synthesis and analysis.
  \item Our work is amenable to stability analysis.
  \item In the car-following example, our approach maintains physically interpretable representations of the underlying dynamics, \emph{e.g.} state-space variables from first principles.
\end{itemize}

In \Cref{Section:RelatedWork}, we provide a context to our contribution by reviewing related work. \Cref{Section:Notation} covers notational remarks. \Cref{Section:ProblemFormulation} introduces the problem statement concisely in the context of an application domain, while \Cref{Section:MainResult} presents the main result. In \Cref{Section:Conclusion} we provide conclusions and future directions. \Cref{Section:SimulationResults} shows simulations using CARLA from \cite{Dosovitskiy17}.

\subsection{Related Work} \label{Section:RelatedWork}
Several recent results for vision-in-the-loop control attempt to leverage learning-based approaches via end-to-end learning, mainly imitation learning, to essentially map pixels to actions via a static map as in \cite{DBLP:journals/corr/BojarskiTDFFGJM16} and \cite{8594386} in the context of driving. Another body of work attempts to first get a latent representation of the underlying dynamics of the process from visual input as in \cite{NIPS2015_5964}, \cite{pmlr-v84-banijamali18a}, \cite{pmlr-v97-hafner19a} and structured latent representations as in \cite{NIPS2016_6379}. In \cite{pmlr-v97-zhang19m}, such latent representations are used in model-based reinforcement learning in the context of manipulation. 

In \cite{5733432}, geometric feature extraction or matching was alleviated by using the luminance of all pixels in 2D direct visual servoing. However, such methods require computing explicitly an interaction matrix and solving a nonlinear optimization problem resulting in a small region of convergence. Therefore, in \cite{7989442} and \cite{8461068}, the relative pose error is learned from a current and reference images for the purpose of posed-based visual servoing. While these methods alleviate the need for camera parameters and scene geometry, servoing is done towards a non-moving target.

In \cite{8957584}, a data-driven simulator is used to train a policy via reinforcement learning from an initial stable policy provided by a human driver. The simulator generates perturbations along an initial policy by taking a 2D image captured along the initial trajectory, creating a depth map and a 3D point cloud, applying a desired viewpoint transformation on the 3D data, then synthesizing a novel 2D view of the scene. The view synthesizer is not deployed at inference time; only the learned policy is.

A different body of work leverages video prediction in the form of visual foresight and scene view synthesis \cite{8594031} and \cite{8624332} along with model predictive control as in \cite{8750823} in the context of robot navigation.

Closed-loop stability is emphasized in \cite{6736325} in the context of sloshing dynamics, where no geometric feature extraction is done. Instead, a single-input multi-output system identification is used to approximate and map linearly the input to a matrix representing a reduced grayscale image of the liquid surface. The linear time-invariant (LTI) system is then converted to a port-Hamiltonian system where a passivity-based controller is applied. To reduce computational intensity, \cite{7039720} applies model reduction on the matrix space to reduce output size then performs LQG control, while \cite{6781580} uses feature extraction to map the liquid surface to polynomial space.

Another recent approach by \cite{pmlr-v120-dean20a} proposes to learn a perception map from high-dimensional data, the image, to a low dimensional latent representation as a state or partial state observation. Robust control is applied on the low dimensional latent representation, resulting in a dynamic output feedback controller and stability is shown under specific conditions, and extended by \cite{dean2020certainty} and \cite{dean2020guaranteeing} to show safety.

In \cite{suh2020surprising}, Lyapunov stability to a target set is shown for an approach based on image visual foresight using linear models to solve a quasi-static pile manipulation problem. The state represents a grayscale image of the pile and image-to-image transitions are learned via switched-linear models. The action space is discrete and switching among actions corresponds to switching among linear models.

\subsection{Notation} \label{Section:Notation}
$\mathbb{R}$ denotes the real line. Given multidimensional array $Y\in \mathbb{R}^{p\times q \times r} $,  $vec(\cdot )$ orderly stacks the $q\times r$ columns of $Y$ one slice at a time until $r$. An all one-entries $n\times m$ matrix is denoted by $\mathbf{1}_{n\times m} $, and by $\mathbf{1}$ when the size is context-dependent. A continuous vector function $f$ of dimension $m$ that is a function of an $n$ dimensional vector is represented by $\mathcal{C}(\mathbb{R}^{n\times 1},\mathbb{R}^{m\times 1})=\{f:\mathbb{R}^{n\times 1} \rightarrow \mathbb{R}^{m\times 1} | f \in \mathcal{C}\}$. $I_{Cam}\in \mathbb{R}^{W\times H \times C}$ denotes an RGB image from a camera, and $I_{Syn}\in \mathbb{R}^{W\times H \times C}$ is an RGB image from a synthesizer, of width, height and channel sizes denoted by $W$, $H$ and $C$ respectively.

\section{Problem Formulation --- Car-Following} \label{Section:ProblemFormulation}

\begin{figure}[!htb]
\centering
\begin{minipage}{\textwidth}
\begin{minipage}{0.35\textwidth}
\hspace*{-3ex}
\tikz{
    \node[anchor=south west,inner sep=0] (image) at (0,0,0) {\includegraphics[width=1.8in]{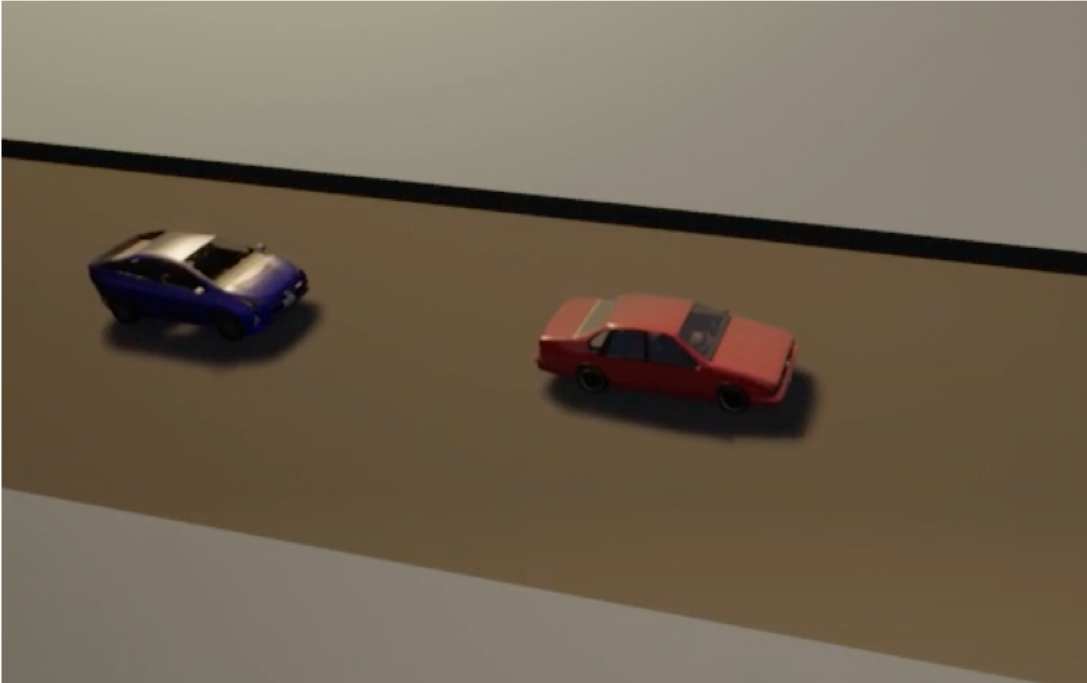}};
    \begin{scope}[x={(image.south east)},y={(image.north west)}]
        \node[align=center,text width=4cm, text=white] (c) at (0.35,0.66) {\small $v_2$};
        \node[align=center,text width=4cm, text=white] (c) at (.8,0.56) {\small $v_1$};
        \node[align=center,text width=4cm, text=white] (c) at (.35,0.4) {\small $s$};
        \draw[white,ultra thick,-latex] (0.29,0.60)  --  (0.4,0.57);
        \draw[white,ultra thick,-latex] (0.74,0.50)  --  (0.85,0.47);
        \draw[white,line width=2pt,dashed] (0.25,.49) --  (0.5,.42);
    \end{scope}
}
\caption{Car-following.}
\label{Figure:CarFollowing}
\end{minipage}
\hspace*{-4ex}
\begin{minipage}{0.34\textwidth}
\begin{tabular}{ll}
$v_1(t)\in \mathbb{R}$ & leader speed, \\
$v_2(t)\in \mathbb{R}$& follower speed,\\
${f}_2(t)\in \mathbb{R}$ & follower force,\\
$s(t)\in \mathbb{R}$ & spacing,\\
$\bar{v}$ & leader desired speed, \\
$\bar{s}$ & desired spacing,\\
$\bar{f}_2$ & steady-state force.\\
\end{tabular}
\end{minipage}
\hspace*{1ex}
\begin{minipage}{0.30\textwidth}
\begin{tabular}{l}
error signals:\\
$x_1(t) = \tilde{v}_1(t) = \bar{v}(t) - {v}_1(t)$,\\
$x_3(t) = \tilde{v}_2(t) = \bar{v}(t)-{v}_2(t) $,\\ 
$x_2(t) = \tilde{s}(t) = \bar{s} - s(t)$,\\ 
$u(t) =  \tilde{f_2}(t)= \bar{f}_2(t) - f_2(t)$,\\
$m_1, m_2>0$ mass of vehicles,\\
$\alpha_1,\alpha_2>0$ drag coefficients.\\
\end{tabular}
\end{minipage}

\end{minipage}
\end{figure}
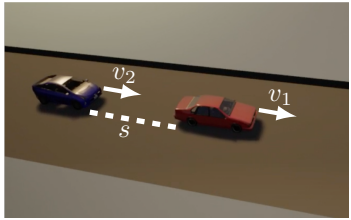

We formulate the problem in the context of a concrete example from the application domain of autonomous driving, namely car-following as depicted in \Cref{Figure:CarFollowing}. In this case, the objective is for the autonomous blue car to follow a leading red car by matching its speed and keeping a desired longitudinal inter-vehicle spacing. The error dynamics can be written as follows:

\begin{subequations} \label{Equation:CarFollowingDynamics}
	\begin{tabular}{p{7cm}p{7cm}}
	    \begin{equation}
		    \dot{x}_1(t) = -\tfrac{\alpha_1}{m_1} x_1(t), \label{Equation:Leader}
		\end{equation} & 
		\begin{equation}
		\dot{x_2}(t) = x_1(t) - x_3(t), \label{Equation:Spacing}
		\end{equation} \\[-27pt]
		\begin{equation}
		\dot{x}_3(t) = -\tfrac{\alpha_2}{m_2} x_3(t) + \tfrac{1}{m_2} u, \label{Equation:Follower}
		\end{equation} &
		\begin{equation}
		y(t) = I_{Cam}(\bar{s}-x_2,\Theta,\Omega), \label{Equation:MeasurementModel}
		\end{equation}\\[-18pt]
	\end{tabular}
	\begin{equation}
    	    e(t) = \bar{y} - y = I_{Cam}(\bar{s},\Theta,\Omega) -I_{Cam}(\bar{s}-x_2,\Theta,\Omega). \label{Equation:TrackingError}
    \end{equation}
\end{subequations}
 
\Crefrange{Equation:Leader}{Equation:Follower} follow from \cite{1098376}. \Cref{Equation:MeasurementModel} is a measurement model where $I_{Cam}(s,\Theta,\Omega)$ represents an image captured by a front-facing camera attached to the follower, and where $\Theta$ represents specific parameters of the leader, while $\Omega$ represents specific parameters of the driving environment background. Moreover, $I_{Cam}(\bar{s},\Theta,\Omega)$ in \eqref{Equation:TrackingError} is a reference image for the same $\Theta$ and $\Omega$ had the spacing been the desired spacing $\bar{s}$. In some sense $I_{Cam}(\bar{s},\Theta,\Omega)$ can be thought of as imagined instead of measured unlike $I_{Cam}(s,\Theta,\Omega)$ which is measured. This builds on neuroscientific concepts of \emph{analysis-by-synthesis} where it is believed that mental imagery plays a role in human vision \cite{Yildirimeaax5979}. We show in \Cref{Section:ViewSynthesis} how to obtain $I_{Cam}(\bar{s},\Theta,\Omega)$.

\begin{assumpn} \label{Assumption:ReferenceImageInvariance}
    \emph{Background Invariance}: Assume that 
    \begin{equation} \label{Equation:ReferenceImageInvariance}
        e(t) = I_{Cam}(\bar{s},\Theta,\Omega) - I_{Cam}(\bar{s}-x_2,\Theta,\Omega) = H(x_2,\bar{s},\Theta).
    \end{equation}
    This says that $e(t)$ is invariant to background changes $\Omega$.
\end{assumpn}

\begin{assumpn} \label{Assumption:Kernel}
    \emph{Null Space}: For $H(x_2,\bar{s},\Theta)$ in \eqref{Equation:ReferenceImageInvariance}, assume that $H(x_2,\bar{s},\Theta) = \mathbf{0} \iff x_2 = 0$. Then it follows that for a given $\bar{s},\Theta$
    \begin{equation} \label{Equation:OutputKernel}
        ker(e(x))=\{x\in \mathbb{R}^{n\times 1}: x_2=0\}.
    \end{equation}
\end{assumpn}

\begin{assumpn} \label{Assumption:ErrorDirection}
    \emph{Error Direction}: Let $h(x_2,\bar{s},\Theta) = vec(H(x_2,\bar{s},\Theta))$. Without loss of generality,
    \begin{equation} \label{Equation:ErrorDirection}
        \bar{s} - s\geq 0  \iff \mathbf{1}^{\intercal} \cdot h(\bar{s} - s,\bar{s},\Theta) \geq 0.
    \end{equation}
\end{assumpn}

\begin{assumpn} \label{Assumption:ReferenceImageMonotonic}
    \emph{Monotonic}: For a given $\bar{s}$ and $\Theta$, consider $h(x_2,\bar{s},\Theta)$ in \eqref{Equation:ErrorDirection}. If $\beta \geq \alpha \geq 0$ or $-\beta \geq -\alpha \geq 0$, then
    \begin{equation} \label{Equation:ReferenceImageMonotonic}
        h(\beta,\bar{s},\Theta)^{\intercal} h(\beta,\bar{s},\Theta) \geq h(\alpha,\bar{s},\Theta)^{\intercal} h(\alpha,\bar{s},\Theta).
    \end{equation}
\end{assumpn}

\begin{assumpn} \label{Assumption:LocallyQuadratic}
    \emph{Locally Quadratic}: For a given $\bar{s}$ and $\Theta$, consider $h(x_2,\bar{s},\Theta)$ in \eqref{Equation:ErrorDirection}. We assume that over a local domain $D \subset \mathbb{R}^{n\times 1}$, where $x=\mathbf{0} \in D$, that
    \begin{equation} \label{Equation:LocallyQuadratic}
        h(x_2,\bar{s},\Theta)^{\intercal} h(x_2,\bar{s},\Theta) \approx c^2(\bar{s},\Theta) {x_2}^2,
    \end{equation}
    for some \emph{nonzero} constant $c(\bar{s},\Theta) \in \mathbb{R}$.
\end{assumpn}

\begin{defn} \label{Definition:UUB} 
	\emph{Uniformly Ultimately Bounded (UUB)} \cite{khalil2002nonlinear}: A solution of $\dot x(t) = f(t,x)$ is said to be UUB with an ultimate bound of $\epsilon$ if $\exists \epsilon >0, \Delta >0$ such that $\forall \delta\in (0,\Delta)$, $\exists T(\delta,\epsilon) \geq 0:$ 
	\[\lVert x(t_0) \rVert_{2} \leq \delta \implies  \lVert x(t) \rVert_{2} \leq \epsilon, \forall t\geq t_0 +T(\delta,\epsilon).\]
\end{defn}

\begin{prb} \label{Problem:OutputRegulationAsymptotic}
	\emph{Output Regulation Solvability}: Consider the car-following dynamics \eqref{Equation:CarFollowingDynamics}. Determine the existence of a static policy
	\begin{equation}
        u = F(y,\bar{y}), \label{Equation:StaticPolicy}
    \end{equation}
    such that the regulated output $vec(e(t))$ is asymptotically stable with $x_1$, $x_2$, and $x_3$ bounded.
\end{prb}

Note that the static control policy \eqref{Equation:StaticPolicy} does not require knowledge of $\bar{v}$.

\begin{prb} \label{Problem:OutputRegulationUUB}
	\emph{Learning-Based Output Regulation}: Find a static policy \eqref{Equation:LearnedStaticPolicy} for the dynamics \eqref{Equation:CarFollowingDynamics}, where $\hat{\bar{y}}$ is learned to approximate $\bar{y}$, such that $vec(\hat{e}(t))= vec(\hat{\bar{y}}-y) $, $x_1$, $x_2$, and $x_3$ are UUBs:
	\begin{equation}
        u = F(y,\hat{\bar{y}}). \label{Equation:LearnedStaticPolicy}
    \end{equation}
\end{prb}

\section{Main Result} \label{Section:MainResult}
In \Cref{Section:Existence} we show the existence of solutions to \Cref{Problem:OutputRegulationAsymptotic} by casting the problem as a static output feedback problem. In \Cref{Section:ViewSynthesis}, we show the architecture of a view synthesizer that will be used to provide reference images needed to compute the tracking error and regulated output \eqref{Equation:TrackingError}. In \Cref{Section:BlockDiagram}, we show a block diagram of the proposed controller, and discuss how to treat the RGB values so that generality is not lost as stated in \Cref{Assumption:ErrorDirection}. Later in \Cref{Section:Stability}, we show closed-loop stability with the camera and the reference view synthesizer in the loop.

\subsection{Existence of Solutions} \label{Section:Existence}
To address \Cref{Problem:OutputRegulationAsymptotic}, we first note that \eqref{Equation:StaticPolicy} is a static policy. One direction to follow is therefore to reduce \Cref{Problem:OutputRegulationAsymptotic} into the following problem.

\begin{prb} \label{Problem:SOF}
	\emph{Static Output Feedback}: Consider the car-following dynamics \eqref{Equation:CarFollowingDynamics}. Determine the existence of a static policy \eqref{Equation:SOF} such that $x_1$, $x_2$, and $x_3$ are asymptotically stable:
	\begin{equation}
        u(t) = F(e(t)). \label{Equation:SOF}
    \end{equation}
    
\end{prb}

\Cref{Problem:SOF} is a state-regulation problem. The next theorem shows the existence of a solution to this state-regulation \Cref{Problem:SOF}, and thus to the output regulation \Cref{Problem:OutputRegulationAsymptotic}.

\begin{lem} \label{Lemma:SOF}
For a fixed $\bar{s}$, $\Theta$, consider writing \eqref{Equation:Leader}, \eqref{Equation:Spacing} and \eqref{Equation:Follower} in the form $\dot{x}=f(x)+g(x)u(x)$ and $h(x)=h(x_2,\bar{s},\Theta)$. There exists $V(x)=x^{\intercal}Px$ with $P=P^{\intercal}\geq \mathbf{0}$ and $G(x) \in \mathcal{C}(\mathbb{R}^{n\times 1},\mathbb{R}^{m\times 1})$ such that over a domain $D \subset \mathbb{R}^{n\times 1}$, where $x=\mathbf{0} \in D$:
\begin{subequations} \label{Equation:HJBCoupled}
    \begin{align}
        0 =& \frac{dV(x)}{dx}^{\intercal}f(x)-\frac{1}{4}\frac{dV(x)}{dx}^{\intercal}g(x)g^{\intercal}(x)\frac{dV(x)}{dx} + h^{\intercal}(x)h(x)+G^{\intercal}(x)G(x), \label{Equation:HJB}\\
        0 =& \frac{dV(x)}{dx}^{\intercal}f(x), \hspace{20pt}  \forall x \in ker(h(x)). \label{Equation:HJBKernel}
    \end{align}
\end{subequations}
\end{lem}

\begin{thm} \label{Theorem:SOF}
A static output feedback policy \eqref{Equation:SOF} exists that solves \Cref{Problem:SOF}, and thus \Cref{Problem:OutputRegulationAsymptotic}.
\end{thm}

\begin{proof} First, if \Cref{Problem:SOF} has a solution, this implies that \Cref{Problem:OutputRegulationAsymptotic} is solvable because $u(t)=F(y,\bar{y})=F(e(t))$ and $\lim_{t\to\infty} x_2(t) = 0 \implies \lim_{t\to\infty} H(x_2(t),\bar{s},\Theta) = \mathbf{0}$ by  \Cref{Assumption:Kernel} and local continuity from \Cref{Assumption:LocallyQuadratic}. From \Cref{Lemma:SOF} there exists a positive semi-definite solution to \eqref{Equation:HJBCoupled} over a domain $D \subset \mathbb{R}^{n\times 1}$. It follows from \cite{1137557} and \cite{10.1007/BFb0110207} that there exists a stabilizing state-feedback policy 
\begin{equation}
    u(x) = G(x)-\frac{1}{2}g^{\intercal}(x)\frac{dV(x)}{dx}, \label{Equation:SOFMap} \\
\end{equation}

\noindent and using the rank theorem, \eqref{Equation:SOFMap} can be written as a static output feedback policy $u(h(x)) = F(e(t))$ over a region around the equilibrium point. Thus \Cref{Problem:SOF} has a solution.
\end{proof}

\subsection{Reference View Synthesis} \label{Section:ViewSynthesis}
We show how to synthesize an imagined reference image $\bar{y} = I_{Cam}(\bar{s},\Theta,\Omega)$ that places the leading car at the desired inter-vehicle spacing $\bar{s}$ as would be viewed by the following car.  In doing so, $\bar{y}$ needs to ideally satisfy \Cref{Assumption:ReferenceImageInvariance}. To do so, we consider an approach based on appearance flow \cite{10.1007/978-3-319-46493-0_18} which has been proposed in the context of 3D view transformation \cite{10.1007/978-3-319-46478-7_20}. \pagebreak However, our objective herein is not to transform the entire view, but rather to generate a view that corresponds to moving an object in the scene closer to or farther away from the observer through a frozen background. Moreover, unlike other  work subsequent to \cite{10.1007/978-3-319-46493-0_18}, namely \cite{8099565}, we do not worry about occlusion issues that are more relevant in the rotation of 3D objects and the need to inpaint the hidden sides of the object by hallucinating a view completion.

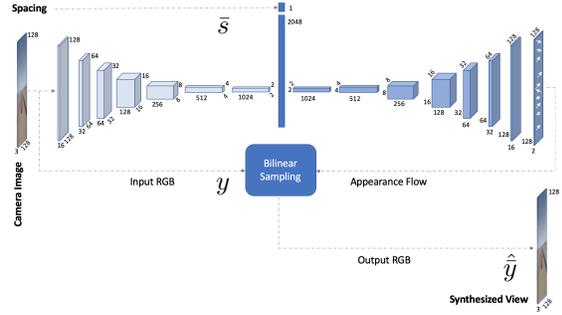
\begin{wrapfigure}{r}{0.50\textwidth}
\centering
\tikz{
    \node[anchor=south west,inner sep=0] (image) at (0,0,0) {\includegraphics[width=3in]{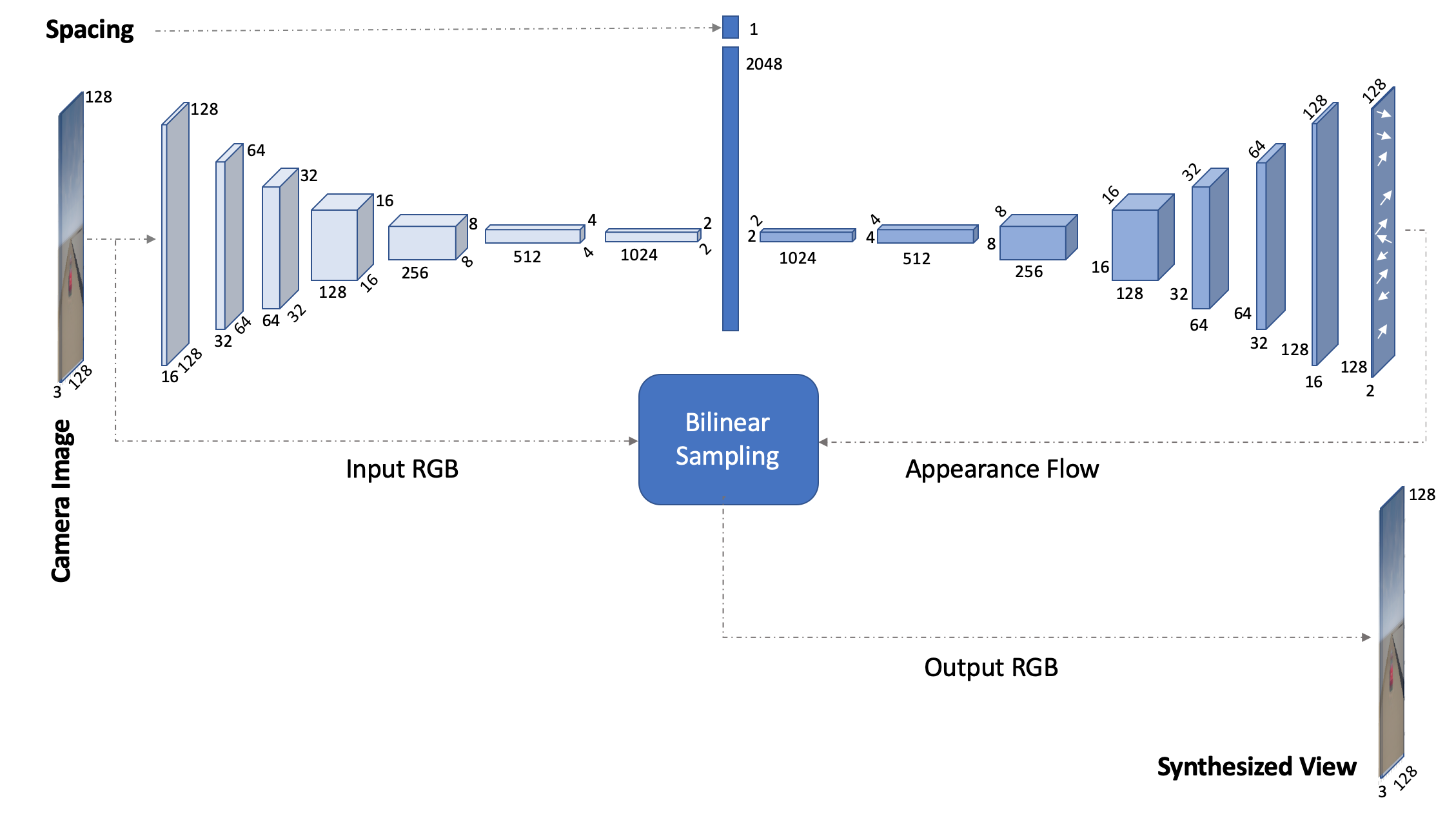}};
    \begin{scope}[x={(image.south east)},y={(image.north west)}]
        \node[align=center,text width=4cm, text=black] (c) at (0.4,0.91) {\small $\bar{s}$};
        \node[align=center,text width=4cm, text=black] (c) at (0.4,0.41) {\small $y$};
        \node[align=center,text width=4cm, text=black] (c) at (0.9,0.16) {\small $\hat{\bar{y}}$};
    \end{scope}
}
\caption{Reference View Synthesizer.} \label{Figure:NeuralNetArchitecture}
\end{wrapfigure}

Our reference view synthesizer is shown in \Cref{Figure:NeuralNetArchitecture} which takes as input a raw camera RGB image and the desired inter-vehicle spacing, and generates as output a view placing the leading vehicle at the desired spacing away from the following vehicle. The raw camera image is an input to an autoencoder that is trained to generate an appearance flow as its output based on $\bar{s}$, where this appearance flow determines which pixels from the camera raw image to copy from as opposed to generating pixels from scratch. The generated appearance flow and the raw camera image are both fed to a bilinear sampler and the output is an RGB image representing the synthesized view. Note that the camera raw image is an input to \emph{both} the autoencoder, and the bilinear sampler. The bilinear sampler is differentiable for backpropagation purposes as shown in \cite{NIPS2015_33ceb07b}.

The encoder is constructed from 8 convolutional neural networks (CNNs) each followed by a rectified linear unit (RELU) and with the last layer flattened. All have a stride of 2, padding of 1 and kernel of 4 except for the first layer which has a kernel size of 3, stride of 1 and padding of 1.

The decoder is constructed from 7 convolutional transpose neural networks with stride 2, padding 1 and kernel of 4, each followed by a RELU and a CNN with kernel 3, stride 1 and padding 1 followed by a tangent hyperbolic function. The last layer clearly outputs values between -1 and 1, representing the appearance flow. The input to the decoder is the flattened output of the encoder in addition to the desired spacing $\bar{s}$.

The bilinear sampler takes as input the raw camera RGB tensor and the appearance flow-field tensor which acts on an identity sampling grid to form a modified sampling grid. The modified sampling grid determines, for each output pixel, the location of the input pixels to copy from. Almost all background pixels are copied from their original locations as is to ensure background invariance, while the pixels representing the current location of the leading car and the desired location are impacted. The reference image can therefore be represented as
\begin{equation}\label{Equation:Synthesized_Yhat}
    \hat{\bar{y}} = I_{Syn}(\bar{s},I_{Cam}(\bar{s}-x_2,\Theta,\Omega)),
\end{equation}
which will be assumed to satisfy \Crefrange{Assumption:ReferenceImageInvariance}{Assumption:LocallyQuadratic}, and the following assumption.
\begin{assumpn} \label{Assumption:ViewSynthesizerError}
    \emph{View Synthesis Error}: For a given $\Theta$ and $\bar{s}$, $\exists \epsilon_1 >0$ such that
    \begin{equation} \label{Equation:ViewSynthesizerError}
        \lVert vec(\;\underbrace{I_{Cam}(\bar{s},\Theta,\Omega)}_{\bar{y}} -  \underbrace{I_{Syn}(\bar{s},I_{Cam}(s,\Theta,\Omega))}_{\hat{\bar{y}}}\;)\rVert_{2}  \leq \epsilon_1.
    \end{equation}
\end{assumpn}

\subsection{Block Diagram of the Feedback Loop} \label{Section:BlockDiagram}

\begin{figure*}[!htb]
\centering
\tikz{
    \node[anchor=south west,inner sep=0] (image) at (0,0,0) {\includegraphics[width=4.7in]{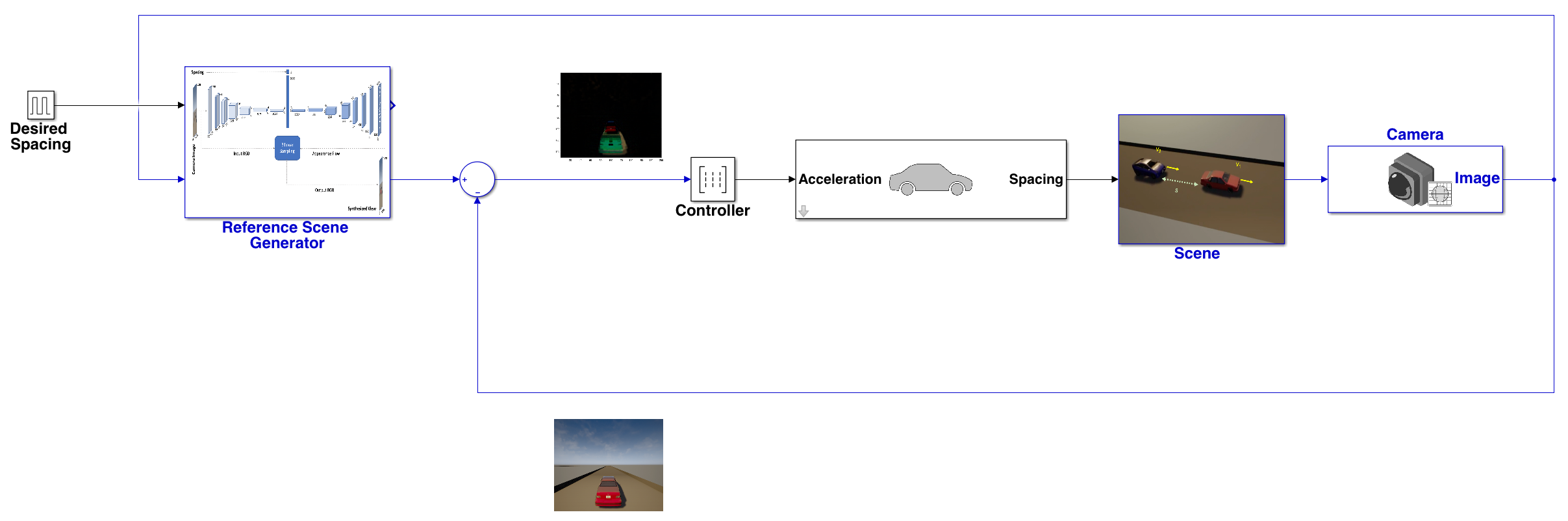}};
    \begin{scope}[x={(image.south east)},y={(image.north west)}]
        \node[align=left,text width=1cm, text=black] (c) at (0.11,0.75) {\small $\bar{s}$};
        \node[align=left,text width=1cm, text=black] (c) at (.11,0.59) {\small $y$};
        \node[align=left,text width=1cm, text=black] (c) at (0.30,0.59) {\small $\hat{\bar{y}}$};
        \node[align=left,text width=1cm, text=black] (c) at (.36,0.59) {\small $\hat{\bar{y}}-y$};
        \node[align=left,text width=1cm, text=black] (c) at (0.525,0.59) {\small $u$};
        \node[align=left,text width=1cm, text=black] (c) at (0.73,0.59) {\small $s$};
        \node[align=right,text width=1cm, text=black] (c) at (0.94,0.59) {\small $y$};
        \node[align=left,text width=1cm, text=black] (c) at (.80,0.2) {\small $y$};
    \end{scope}
}
\caption{Block Diagram of Feedback Loop.} \label{Figure:BlockDiagramPreliminary}
\end{figure*}
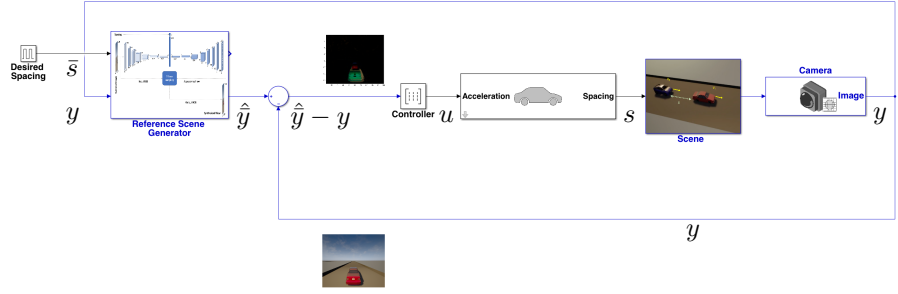

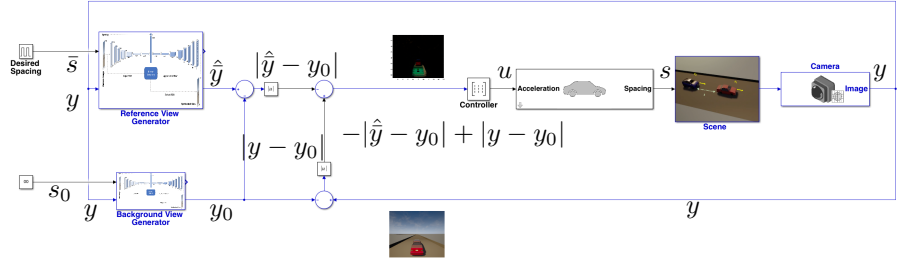
\begin{figure*}[!htb]
\centering
\tikz{
    \node[anchor=south west,inner sep=0] (image) at (0,0,0) {\includegraphics[width=4.7in]{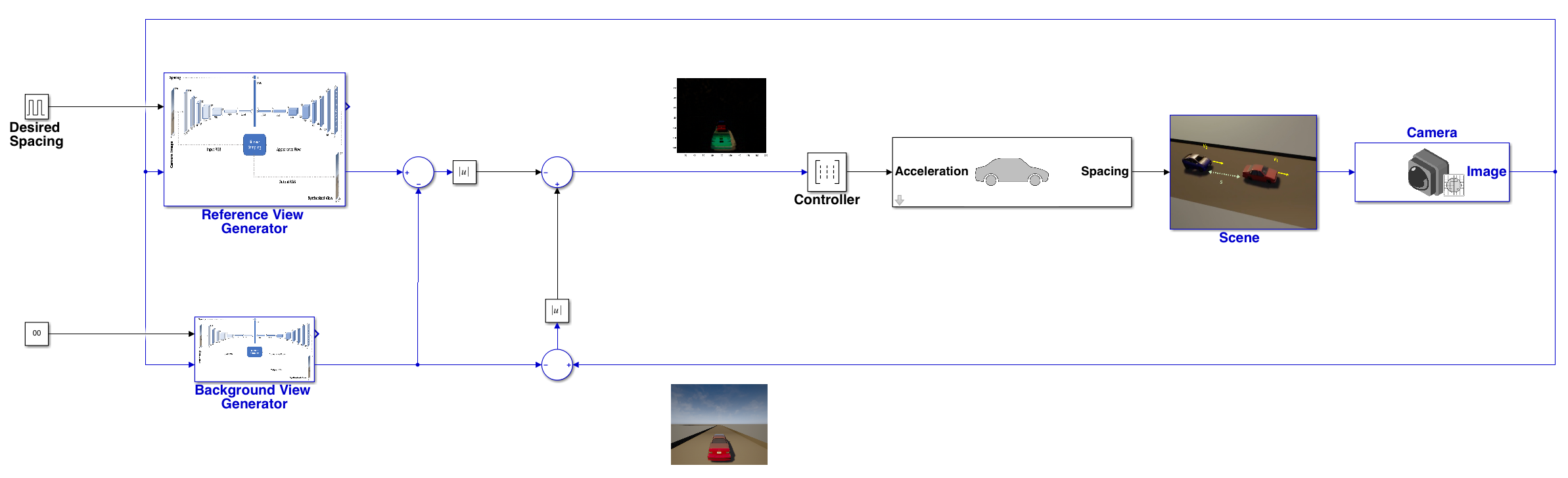}};
    \begin{scope}[x={(image.south east)},y={(image.north west)}]
        \node[align=left,text width=1cm, text=black] (c) at (0.11,0.73) {\small $\bar{s}$};
        \node[align=left,text width=1cm, text=black] (c) at (.11,0.59) {\small $y$};
        \node[align=left,text width=1cm, text=black] (c) at (0.09,0.26) {\small $s_0$};
        \node[align=left,text width=1cm, text=black] (c) at (.13,0.19) {\small $y$};
        \node[align=left,text width=1cm, text=black] (c) at (0.27,0.70) {\small $\hat{\bar{y}}$};
        \node[align=left,text width=1cm, text=black] (c) at (0.27,0.19) {\small $y_0$};
        \node[align=left,text width=1.5cm, text=black] (c) at (.34,0.73) {\small $\lvert \hat{\bar{y}}-y_0 \rvert$};
        \node[align=left,text width=1.5cm, text=black] (c) at (.325,0.43) {\small $\lvert y-y_0 \rvert$};
        \node[align=left,text width=3cm, text=black] (c) at (.50,0.48) {\small $-\lvert \hat{\bar{y}}-y_0 \rvert + \lvert y-y_0 \rvert$};
        \node[align=left,text width=1cm, text=black] (c) at (0.59,0.70) {\small $u$};
        \node[align=left,text width=1cm, text=black] (c) at (0.77,0.70) {\small $s$};
        \node[align=right,text width=1cm, text=black] (c) at (0.94,0.70) {\small $y$};
        \node[align=left,text width=1cm, text=black] (c) at (.80,0.2) {\small $y$};
    \end{scope}
}
\caption{Block Diagram of a Generalized Feedback Loop.} \label{Figure:BlockDiagram}
\end{figure*}

We start by showing a block diagram of the proposed controller with a camera and a reference view synthesizer in the loop as shown in \Cref{Figure:BlockDiagramPreliminary}. 

As straightforward as this may seem, the block diagram of \Cref{Figure:BlockDiagramPreliminary} may loose the generality \Cref{Assumption:ErrorDirection} states. To see this, consider  a case where $\bar{s}-s \geq 0$ and the following two 3-by-3 pixel images where we show a single color channel only, \emph{e.g.} Green:

\begin{subequations}
    \begin{tabular}{@{}*{2}{m{0.48\linewidth}@{}}}
        \begin{equation}\label{Equation:GreenBlackReference} 
            \bar{z}=\left[ \begin{matrix}
        	            B  & O  & B  \\
        	            B  & B  & B  \\
        	            B  & B  & B  \\
        	\end{matrix} \right], 
        \end{equation} &
        \begin{equation} \label{Equation:GreenBlack}
        	z=\left[ \begin{matrix}
        	            B  & B  & B  \\
        	            B  & B  & B  \\
        	            O  & O  & O  \\
        	\end{matrix} \right].
        \end{equation}
	\end{tabular}
\end{subequations}

The reference image $\bar{z}$ has 1 pixel in the first row denoted by the letter $O$ representing the color of an object traversing a background of color denoted by $B$. The image $z$ has 3 pixels in the last row representing the same observed object at a closer distance to the observer thus occupying more pixels. From \eqref{Equation:ErrorDirection}, we get
\begin{equation}\label{Equation:GreenBlackDirection}
    \mathbf{1}^{\intercal} \cdot vec(\bar{z} - z) = (O-B)+3(B-O) =2(B-O).
\end{equation}

\noindent If the object is black moving in a green background, then we have $O=0$ and $B=1$ and thus $\mathbf{1}^{\intercal} \cdot vec(\bar{z} - z) \geq 0$, otherwise if the object is green and moving through a black background, then $O=1$ and $B=0$ and thus $\mathbf{1}^{\intercal} \cdot vec(\bar{z} - z) \leq 0$.

To enforce the generality of \Cref{Assumption:ErrorDirection}, we need an expression that is invariant to the polarity of $(B-O)$, in other words a function of $\lvert B-O \rvert$. Consider a 3-by-3 pixel image $z_0$ representing the background only whose elements are all $B$ values. By adding and subtracting $z_0$ to \eqref{Equation:GreenBlackDirection} and taking absolute values, we get the following
\begin{equation}
    \label{Equation:GreenBlackAbs}
        \mathbf{1}^{\intercal} \cdot vec(-\lvert z_0 - \bar{z} \rvert +\lvert z_0 -z\rvert) = -\lvert B-O \rvert + 3 \lvert B-O \rvert =2 \lvert B-O \rvert, 
\end{equation}

\noindent which is the desired expression. \Cref{Equation:GreenBlackAbs} provides a clear breakdown to how the error signal can achieve the desired error directionality and magnitude. We therefore reorganize the block diagram in \Cref{Figure:BlockDiagramPreliminary} as shown in \Cref{Figure:BlockDiagram} to ensure the generality of \Cref{Assumption:ErrorDirection} is not lost.

Note that we may use the view synthesizer to generate a background by choosing $s_0$ to be a large value, thus the leading car essentially is vanishing from the view.

\subsection{Stability Analysis of the Learning-Based Controller} \label{Section:Stability}
The stability analysis will be discussed for the block diagram of \Cref{Figure:BlockDiagramPreliminary}. We treat the following nonlinear controller which has a proportional gain acting on a neural network based error signal 
\begin{equation}
    u = vec(K)^{\intercal} \cdot vec(\hat{\bar{y}}-y)=vec(K)^{\intercal} \cdot vec(\hat{e}), \label{Equation:ProportionalController}
\end{equation}

\noindent which relates to \eqref{Equation:StaticPolicy} and mainly \eqref{Equation:SOF}; and where $\hat{\bar{y}}-y = I_{Syn}(\bar{s},I_{Cam}(s,\Theta,\Omega)) - I_{Cam}(s,\Theta,\Omega)$. Note that $\hat{\bar{y}}$ reflects an internal model principle. $\hat{e}$ enables background invariance, thus generalization to backgrounds. Generalization and sample efficiency are key performance issues \cite{chen2020robust} and \cite{sax2019midlevel}.

We first note that the dynamical system \eqref{Equation:CarFollowingDynamics} can be decomposed into two subsystems, a stable uncontrollable subsystem governing the dynamics of $x_1(t)$ and a controllable subsystem governing the dynamics of $x_2(t)$ and $x_3(t)$. By decoupling the stable uncontrollable subsystem, we have:
\begin{subequations} \label{Equation:DecoupledCarFollowingDynamics}
	\begin{align}
		\dot{x_2}(t) =& - x_3(t), \label{Equation:DecoupledSpacing}\\
		\dot{x}_3(t) =& -\frac{\alpha_2}{m_2} x_3(t) + \frac{1}{m_2} u, \label{Equation:DecoupledFollower}\\
		y(t) =& I_{Cam}(\bar{s}-x_2,\Theta,\Omega), \label{Equation:DecoupledMeasurementModel}\\
		\hat{e}(t) =& \hat{\bar{y}} - y =I_{Syn}(\bar{s},I_{Cam}(s,\Theta,\Omega)) - I_{Cam}(\bar{s}-x_2,\Theta,\Omega). \label{Equation:DecoupledTrackingError}
	\end{align}
\end{subequations}

The following theorem demonstrates stability and thus addresses \cref{Problem:OutputRegulationUUB}.

\begin{thm} \label{Theorem:UUB}
Consider controller  \eqref{Equation:ProportionalController} and let $K = \mathbf{1}$. The dynamics \eqref{Equation:DecoupledCarFollowingDynamics} for a fixed $\Theta$ is UUB.
\end{thm}

\begin{proof}
We first construct an appropriate Lyapunov function candidate. Let $u^*(x_2) = vec(K)^{\intercal} \cdot vec(H(x_2,\bar{s},\Theta))$. Consider the following positive definite function for subsystem \eqref{Equation:DecoupledCarFollowingDynamics}
\begin{equation}\label{Equation:LyapunovLike}
    V(x_2,x_3) = w_1{x_2}^2 + w_2{x_2}{x_3} + w_3{x_3}^2 + w_4 \int\limits_{0}^{x_2} u^*(z) d{z},
\end{equation}
\noindent where $w_1>0$ is arbitrary, and $w_3>0$, and $w_2$ are chosen appropriately and such that $w_1 x_1^2 + w_2 x_2 x_3 + w_3 x_3^2$ is positive definite in $x_2$ and $x_3$. Moreover $w_4>0$ will be chosen appropriately noting that the integral term is nonnegative due to \Cref{Assumption:ErrorDirection} and $K = \mathbf{1}$.

From \Cref{Assumption:LocallyQuadratic}, $u^*(z)$ is locally continuous in $z$. Differentiating $V(x_2,x_3)$ along the trajectories of
\eqref{Equation:DecoupledCarFollowingDynamics}, we get
\begin{equation}\label{Equation:LyapunovDerivative}
    \begin{aligned}
        \dot{V}(x_2,x_3) =& 2 w_1 x_2 \dot{x}_2 +w_2 \dot{x}_2 x_3 + w_2 x_2 \dot{x}_3 + 2 w_3 x_3 \dot{x}_3 + w_4 \dot{x}_2 u^*(x_2),  \\
        =& \left(-2w_1 - \frac{\alpha_2 }{m_2} w_2\right) x_2 x_3 + \left(-w_2 - 2\frac{\alpha_2 }{m_2} w_3\right) x_3^2 + \frac{w_2}{m_2} x_2 u + \left(2\frac{w_3}{m_2} u - w_4 u^*\right) x_3.
    \end{aligned}
\end{equation}

Adding and subtracting $\frac{w_2}{m_2} x_2 u^*$ to \eqref{Equation:LyapunovDerivative}, we get
\begin{equation}\label{Equation:LyapunovDerivativeAddSubtract}
    \begin{aligned}
        \dot{V}(x_2,x_3) =& \left(-2w_1 - \frac{\alpha_2 }{m_2} w_2\right) x_2 x_3 + \left(-w_2 - 2\frac{\alpha_2 }{m_2} w_3\right) x_3^2 \\&+ \frac{w_2}{m_2} x_2 (u-u^*) + \frac{w_2}{m_2} x_2 u^* + \left(2\frac{w_3}{m_2} u - w_4 u^*\right) x_3.
    \end{aligned}
\end{equation}

Choosing $w_2=-2\frac{m_2}{\alpha_2}w_1$ to cancel the $x_2x_3$ term, and choosing $w_4=2\frac{w_3}{m_2}$ we get
\begin{equation}\label{Equation:LyapunovDerivativeCancel}
    \begin{aligned}
        \dot{V}(x_2,x_3) =& \left(2\frac{m_2}{\alpha_2}w_1 - 2\frac{\alpha_2}{m_2} w_3\right) x_3^2 - 2\frac{w_1}{\alpha_2}x_2 (u-u^*(x_2)) -2\frac{w_1}{\alpha_2}x_2u^*(x_2) \\&+2\frac{w_3}{m_2}x_3(u-u^*(x_2)).
    \end{aligned}
\end{equation}

We finally choose $w_3 > (\frac{m_2}{\alpha_2})^2 w_1$ to force the  coefficient of the first term in the right-hand side of \eqref{Equation:LyapunovDerivativeCancel} to be negative. We therefore write \eqref{Equation:LyapunovDerivativeCancel} as follows
\begin{equation}\label{Equation:LyapunovDerivativeInequalityAbs}
    \begin{aligned}
        \dot{V}(x_2,x_3) \leq & \left(2\frac{m_2}{\alpha_2}w_1 - 2\frac{\alpha_2}{m_2} w_3\right) |x_3|^2 + 2\frac{w_3}{m_2} |x_3| |u-u^*(x_2)| \\&-2\frac{w_1}{\alpha_2}|x_2||u^*(x_2)| +2\frac{w_1}{\alpha_2}|x_2||u-u^*(x_2)|,\\
        =& -|x_3|\left(\left(2\frac{\alpha_2}{m_2} w_3-2\frac{m_2}{\alpha_2}w_1 \right) |x_3|-2\frac{w_3}{m_2} |u-u^*(x_2)| \right) \\&-2\frac{w_1}{\alpha_2}|x_2|\left(|u^*(x_2)| -|u-u^*(x_2)|\right).
    \end{aligned}
\end{equation}

From \Cref{Assumption:ViewSynthesizerError}, it can be shown that $\exists \epsilon_2 >0$ such that $|u-u^*(x_2)| < \epsilon_2$, which when substituted in \eqref{Equation:LyapunovDerivativeInequalityAbs} we get
\begin{equation}\label{Equation:LyapunovDerivativeInequalityAbsEpsilon}
    \dot{V}(x_2,x_3) \leq  -|x_3|\left(\left(2\frac{\alpha_2}{m_2} w_3-2\frac{m_2}{\alpha_2}w_1 \right) |x_3|-2\frac{w_3}{m_2} \epsilon_2 \right) -2\frac{w_1}{\alpha_2}|x_2|\left(|u^*(x_2)| -\epsilon_2\right).
\end{equation}

It can be shown that $\exists r>0$ and a ball $B([x_2,x_3],r)$ around the origin such that if $[x_2,x_3] \notin B([x_2,x_3],r)$ then $\dot{V}(x_2,x_3) \leq 0$.
\end{proof}

\section{Conclusion} \label{Section:Conclusion}
We demonstrated that stable feedback control directly from raw pixels is plausible and promising, and that introduced assumptions hold reasonably well for the application domain considered within a  simulator environment. For further improvements and scalability, we need to investigate approaches to relax strong assumptions and have the theory encompassing of more practical scenarios and different types of motions and tracked objects, and to further provide quantitative and qualitative assessments on generalization and sample complexity. The method generalized well to different driving backgrounds that have not been seen before due to the ability of the synthesizer to be reasonably invariant to background changes. The approach provides a more clear path to apply control theory directly to pixels and establish safe and trustworthy dynamical systems that are more interpretable compared to purely end-to-end learning approaches. The approach can extend to various automatic control applications where a cheap camera sensor can be deployed for motion control.\footnote{Code is available at \nolinkurl{https://github.com/abukhalaf/FeedbackFromPixels_L4DC2021}}

\acks{Toyota Research Institute provided funds to support this work.}

\appendix
\section{Proof of \Cref{Lemma:SOF}} \label{Section:AppendedProofs}
\begin{proof} 
Note that $h(x)=h(x_2,\bar{s},\Theta)$, and therefore from  \Cref{Assumption:Kernel}, it follows that $ker(h(x))=\{x\in \mathbb{R}^{n\times 1}: x_2=0\}$. Moreover, by writing $f(x)=Ax$, $g(x)=B$, and locally $h^{\intercal}(x)h(x) = c^2 {x_2}^2$ from \Cref{Assumption:LocallyQuadratic}, and $G = [G_1, G_2, G_3]$, and where

\begin{subequations}
\vspace{-1em}
    \begin{tabular}{m{5.5cm}m{4cm}m{4cm}}
        \begin{equation}\label{Equation:A_Matrix} 
            A=\left[ \begin{matrix}
        	            -\frac{\alpha_1}{m_1}  & 0  & 0  \\
        	            1  & 0  & -1  \\
        	            0  & 0  & -\frac{\alpha_2}{m_2}  \\
        	\end{matrix} \right], 
        \end{equation} & 
        \begin{equation} \label{Equation:B_Matrix}
        	B=\left[ \begin{matrix}
        	            0 \\
        	            0 \\
        	            \frac{1}{m_2}\\
        	\end{matrix} \right],
        \end{equation} & 
        \begin{equation} \label{Equation:C_Matrix}
        	C=\left[ \begin{matrix} 
        	            0\;\;\; c\;\;\; 0
        	\end{matrix} \right],
        \end{equation}
	\end{tabular}
\end{subequations}

\noindent we can replace the Hamilton-Jacobi (HJ) equation \eqref{Equation:HJB} and \eqref{Equation:HJBKernel} over domain $D \subset \mathbb{R}^{n\times 1}$ with 
\begin{subequations} \label{Equation:RiccatiCoupled}
    \begin{align}
        0 =& A^{\intercal}P+PA-PBB^{\intercal}P + C^{\intercal}C +G^{\intercal}G, \label{Equation:Riccati}\\
        0 =& N(A^{\intercal}P+PA)N, \hspace{20pt}  N=I-C^{\intercal}(CC^{\intercal})^{-1}C. \label{Equation:RiccatiKernel}
    \end{align}
\end{subequations}
From the kernel condition \eqref{Equation:RiccatiKernel}, we have
\begin{equation}\label{Equation:PKernel}
    P=\left[ \begin{matrix}
	            p_{11}                      & \frac{\alpha_1}{m_1}p_{11}    & -\frac{\frac{\alpha_1}{m_1}p11+\frac{\alpha_2}{m_2}p33}{\frac{\alpha_1}{m_1}+\frac{\alpha_2}{m_2}}  \\
	            \frac{\alpha_1}{m_1}p_{11}  & p_{22}                        & -\frac{\alpha_2}{m_2}p_{33} \\
	            -\frac{\frac{\alpha_1}{m_1}p11+\frac{\alpha_2}{m_2}p33}{\frac{\alpha_1}{m_1}+\frac{\alpha_2}{m_2}} & -\frac{\alpha_2}{m_2}p_{33} & p_{33} \\
	\end{matrix} \right].
\end{equation}

From the algebraic Riccati equation \eqref{Equation:Riccati}, we obtain the following for $p_{11}$, $p_{22}$, $p_{33}$ and $G$:

\begin{subequations}\label{Equation:P_G}
    \begin{tabular}{@{}*{2}{m{0.48\linewidth}@{}}}
        \begin{equation}\label{Equation:p11}
            p_{11} = \frac{\lvert c\rvert \alpha_2 + {\lvert c\rvert}^2 \frac{ m_2}{\alpha2} - {\lvert c\rvert}^2 \frac{\frac{m_1}{\alpha1}\frac{m_2}{\alpha2}}{\frac{m_1}{\alpha1}+\frac{m_1}{\alpha1}}}{\lvert c\rvert \frac{\alpha_1 }{\alpha_1 m_2+\alpha_2 m_1}+\frac{{\alpha_1 }^2}{{m_1}^2}}, 
        \end{equation} & 
        \begin{equation}\label{Equation:G1}
            G_1 = -\frac{\alpha_1 m_2 p_{11} + \alpha_2 m_1 p_{33}}{\alpha_1 {m_2}^2 + \alpha_2 m_1 m_2},
        \end{equation} \\[-15pt]
        \begin{equation}\label{Equation:p22}
            p_{22} = \lvert c\rvert \alpha_2 + {\lvert c\rvert}^2 \frac{ m_2}{\alpha_2},
        \end{equation}&
         \begin{equation}\label{Equation:G2}
            G_2 = 0,
        \end{equation} \\[-15pt]
        \begin{equation}\label{Equation:p33}
            p_{33} = \lvert c\rvert \frac{ {m_2}^2}{\alpha_2},
        \end{equation} &
        \begin{equation}\label{Equation:G3}
            G_3 = \lvert c\rvert \frac{ m_2}{\alpha_2} .
        \end{equation}
    \end{tabular}
\end{subequations}

Substituting \eqref{Equation:p11}, \eqref{Equation:p22} and \eqref{Equation:p33} in \eqref{Equation:PKernel}, it follows that the principal minors of \eqref{Equation:PKernel} are nonnegative; hence $P \geq \mathbf{0}$.
\end{proof}

For specific numerical values of $\alpha_1$, $\alpha_2$, $m_1$ and $m_2$, a numerical procedure shown in \cite{KUCERA19951357} and \cite{doi:10.2514/1.16794} can be used to numerically solve \eqref{Equation:RiccatiCoupled}.

\pagebreak
\section{Simulation Results} \label{Section:SimulationResults}

\subsection{Training and Data Sets} \label{Section:Training}

\begin{figure*}[!htb]
    \centering
    \subfigure[Town 3: Red Car]
        {\label{Figure:Town3Red}\includegraphics[width=0.23\textwidth]{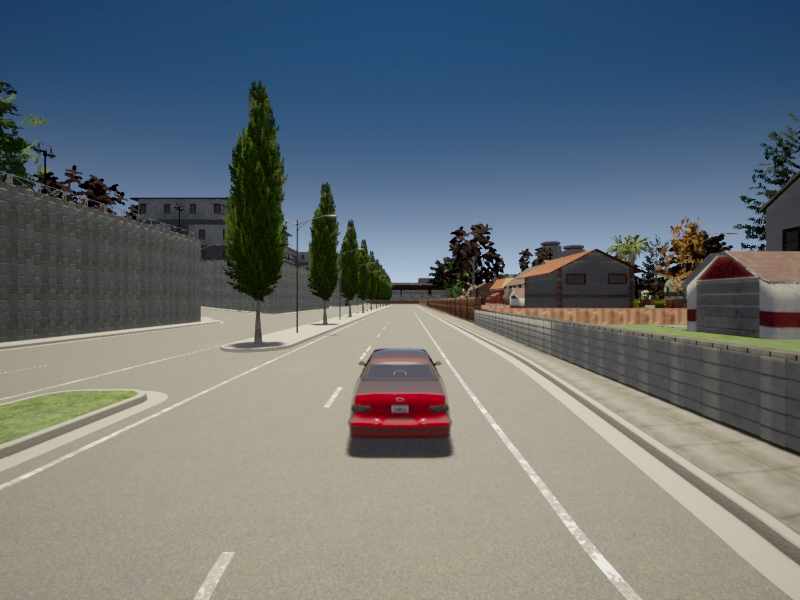}}
    \hfill
    \subfigure[Town 3: Blue Car]
        {\label{Figure:Town3Blue}\includegraphics[width=0.23\textwidth]{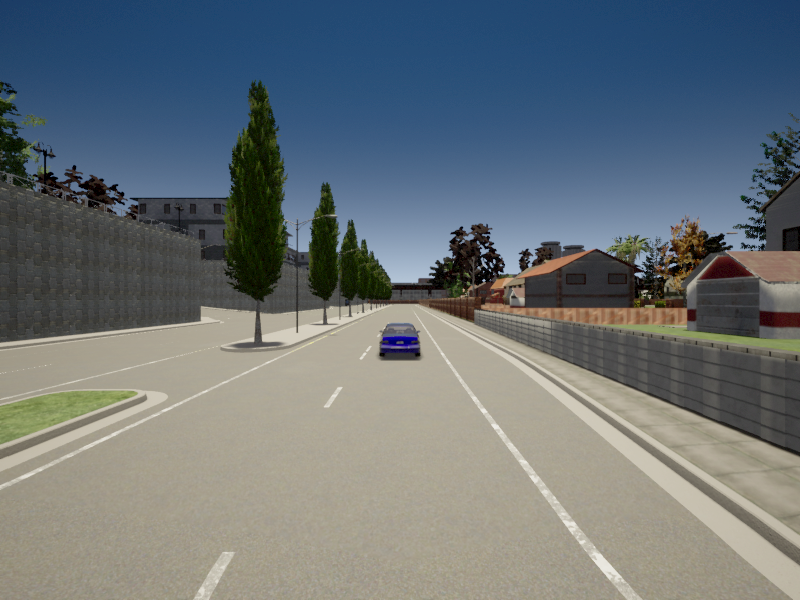}}
    \hfill
    \subfigure[Town 4: Red Car]
        {\includegraphics[width=0.23\textwidth]{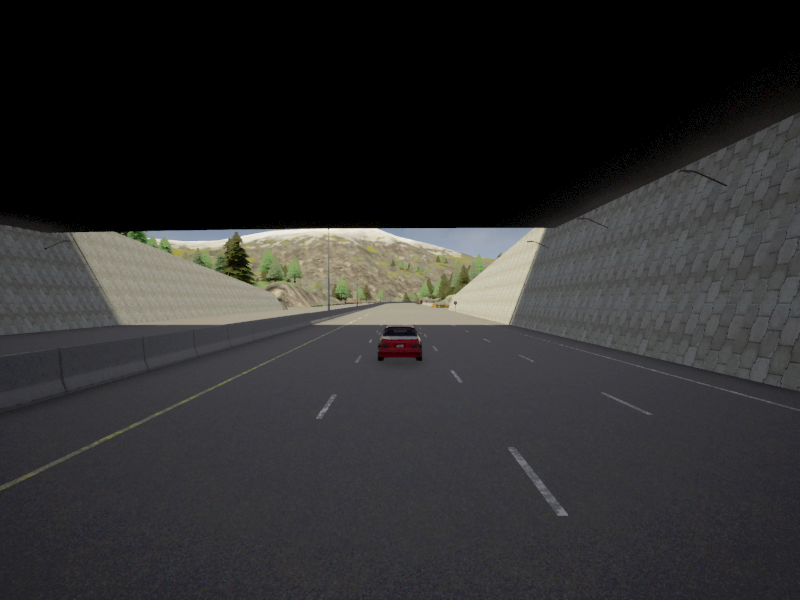}
        \label{Figure:Town4ARed}}
    \hfill
    \subfigure[Town 4: Blue Car]
        {\includegraphics[width=0.23\textwidth]{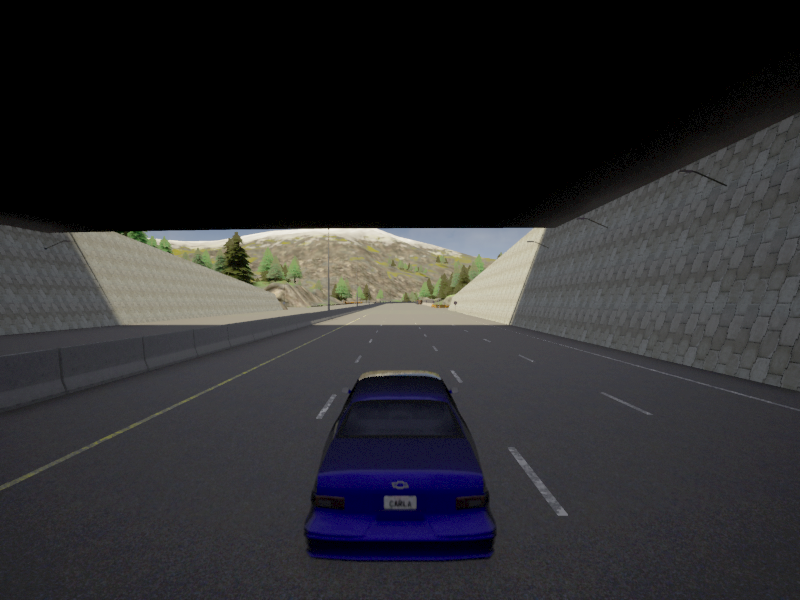}
        \label{Figure:Town4ABlue}}
    \hfill
    \subfigure[Town 4: Red Car]
        {\includegraphics[width=0.23\textwidth]{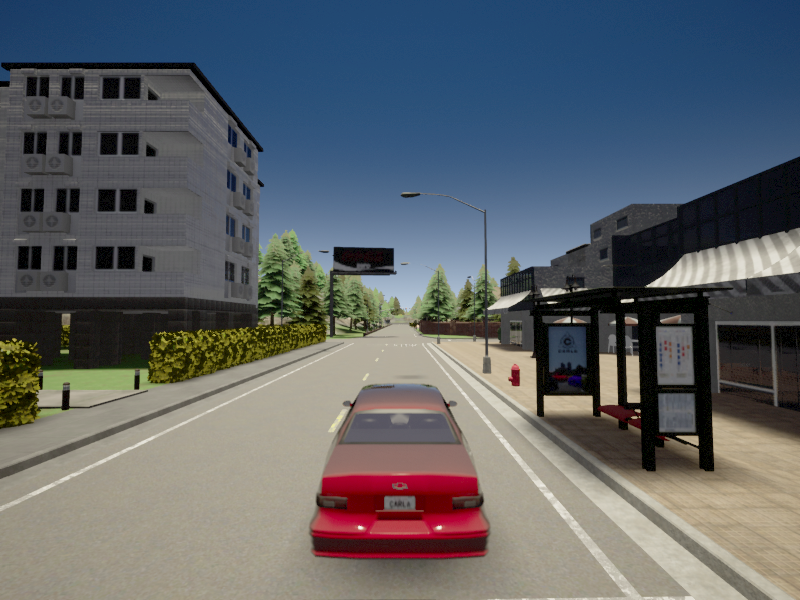}
        \label{Figure:Town4BRed}}
    \hfill
    \subfigure[Town 4: Blue Car]
        {\includegraphics[width=0.23\textwidth]{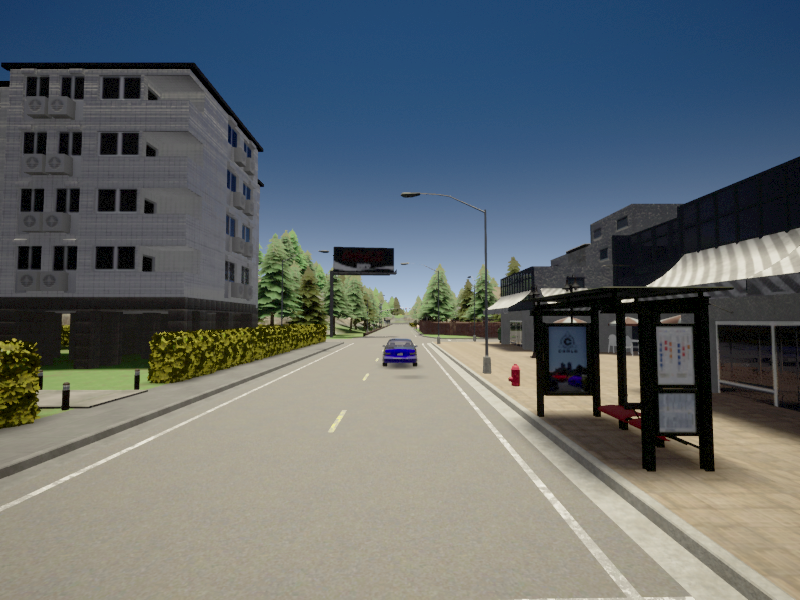}
        \label{Figure:Town4BBlue}}
    \hfill
    \subfigure[Town 5: Red Car]
        {\includegraphics[width=0.23\textwidth]{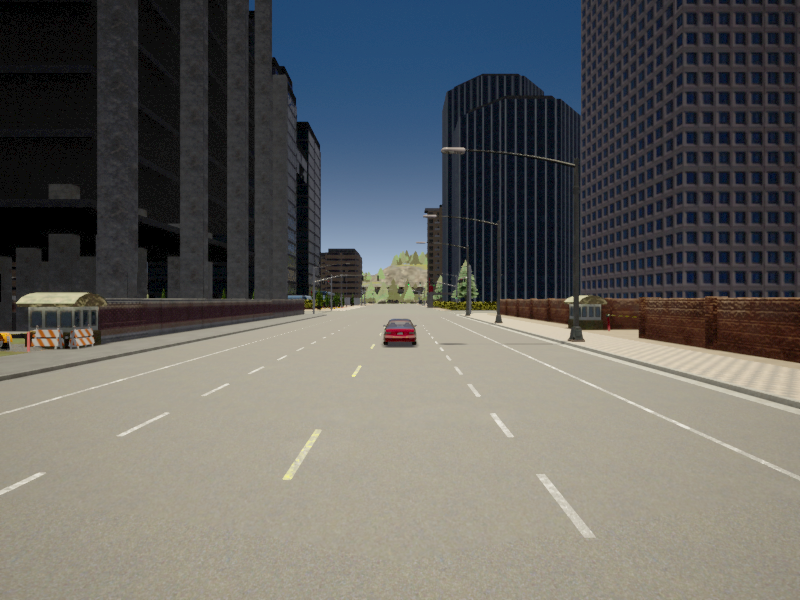}
        \label{Figure:Town5Red}}
    \hfill
    \subfigure[Town 5: Blue Car]
        {\includegraphics[width=0.23\textwidth]{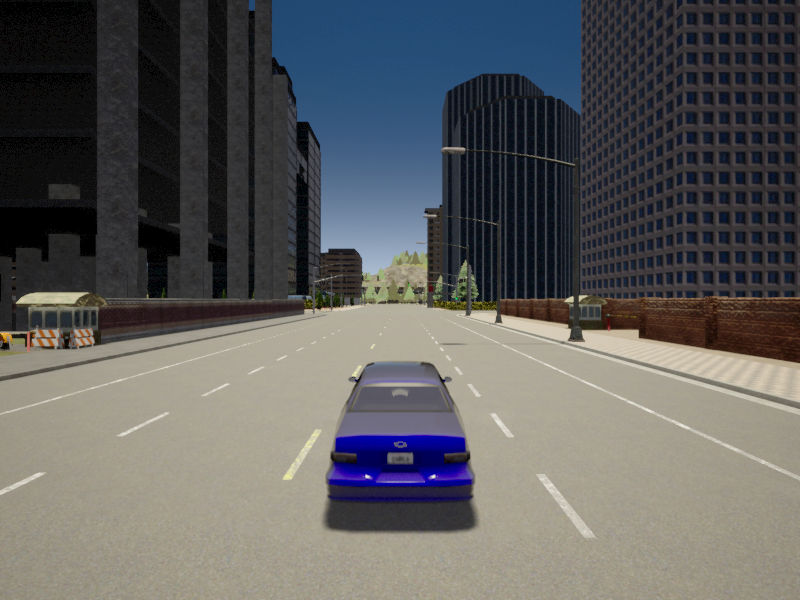}
        \label{Figure:Town5Blue}}
    \hfill
    \caption{Observed Camera Views from three CARLA Towns.}
    \label{Figure:ObservationViews}
\end{figure*}

We use release 0.9.9 of CARLA \cite{Dosovitskiy17}, a photorealistic urban driving simulator, in this study to both create a dataset to use for training the view synthesizer proposed in \Cref{Section:ViewSynthesis}, and to demonstrate our proposed feedback control strategy for the car-following scenario introduced in \Cref{Section:ProblemFormulation}. Our dataset is a set of raw images for observation views along with associated spacing distances collected at four different streets in three different CARLA towns or maps. Namely, one street in Town 3, two streets in Town 4, and one street in Town 5, all representing different urban environments. 
At each street, our data collection methodology is as follows:
\begin{itemize}[noitemsep]
    \item Spawn a leader and a follower cars.
    \item Place a front-facing camera on the follower car that faces the back of the leading car.
    \item Freeze the background, namely weather conditions, sun movement, cloud motion, wind or trees movements, traffic lights, and other agents if any. The only thing allowed to move in the frame is the leading car.
    \item With the follower vehicle fully stopped at speed \SI{0}{\meter\per\second}, record the leading car driving away at an arbitrary slow speed.
    \item Capture images continuously, along with the associated distance, covering an inter-vehicle distance of \SI{5}{\meter} to \SI{50}{\meter}, with one or two samples or images for each driven meter.
    \item Repeat the experiment for the same location but with a different color for the leading car. Data for both blue and red colors is gathered.
\end{itemize}

The neural network is trained in a supervised manner. The network takes two inputs --- a scalar and an RGB image, and it generates one output --- an RGB image. The inputs are the reference distance $\bar{s}$ and an RGB image from the front-facing camera showing the leading vehicle at some distance. The output is trained to generated a synthesized view RGB image that shows the leading vehicle at the desired $\bar{s}$ in the same background of the input RGB image.

The training set is organized in the form of 3-tuples. For each desired spacing $\bar{s}$, there are eight groups with each group corresponding to one of the streets and one of the leading vehicle colors shown in \Cref{Figure:ObservationViews}. Within each group, a 3-tuple has the following items:
\begin{enumerate}[noitemsep]
    \item A desired spacing $\bar{s}$.
    \item A captured RGB image showing the leading car at the desired $\bar{s}$ in the street.
    \item A captured RGB image for the \emph{same} leading car of the \emph{same} color at a distance between \SI{5.5}{\meter} and \SI{40}{\meter} in the \emph{same} street, and therefore the same background.
\end{enumerate}

Within each group, enough 3-tuples are created such that the third item covers the distance \SI{5.5}{\meter} to \SI{40}{\meter} at a \SI{1}{\meter} increment. Our training set considers three different values for $\bar{s}$, namely $\bar{s}=10$, $\bar{s}=20$ and $\bar{s}=30$. Therefore, we have a total of 24 groups of 3-tuples --- 8 groups per each value of $\bar{s}$ --- for a total of 774 different 3-tuples.

\subsection{Open-Loop Results} \label{Section:Open-Loop}

\begin{figure*}[!htb]
    \centering
    \subfigure[Camera View]
        {\includegraphics[width=0.23\textwidth]{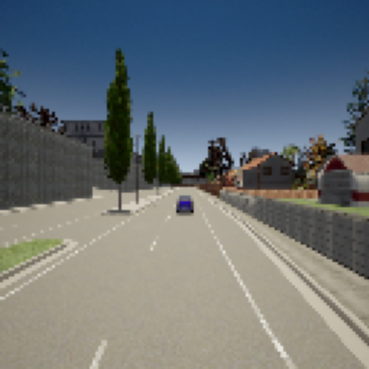}
        \label{Figure:Town3Input}}
    \hfill
    \subfigure[Camera View]
        {\includegraphics[width=0.23\textwidth]{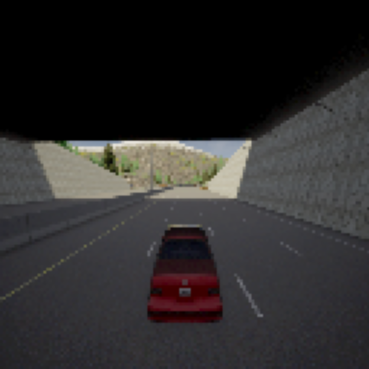}
        \label{Figure:Town4AInput}}
    \hfill
    \subfigure[Camera View]
        {\includegraphics[width=0.23\textwidth]{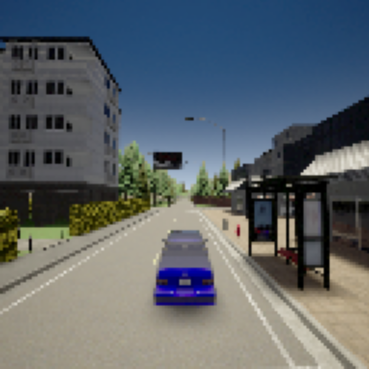}
        \label{Figure:Town4BInput}}
    \hfill
    \subfigure[Camera View]
        {\includegraphics[width=0.23\textwidth]{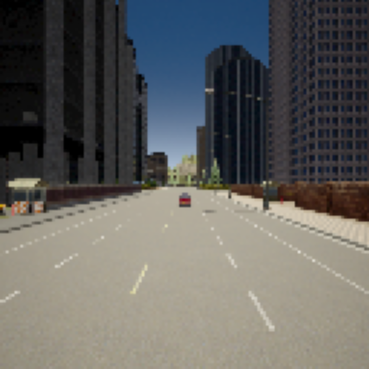}
        \label{Figure:Town5Input}}
    \hfill
    \subfigure[Synthesized View for \Cref{Figure:Town3Input}]
        {\includegraphics[width=0.23\textwidth]{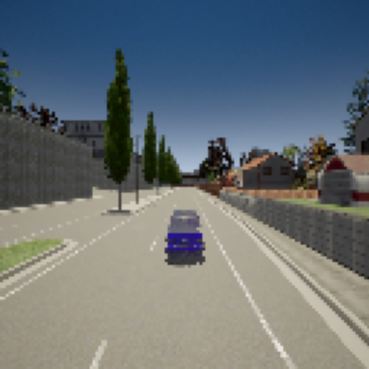}
        \label{Figure:Town3Output}}
    \hfill
    \subfigure[Synthesized View for \Cref{Figure:Town4AInput}]
        {\includegraphics[width=0.23\textwidth]{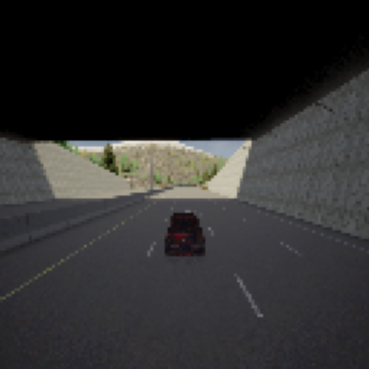}
        \label{Figure:Town4AOutput}}
    \hfill
    \subfigure[Synthesized View for \Cref{Figure:Town4BInput}]
        {\includegraphics[width=0.23\textwidth]{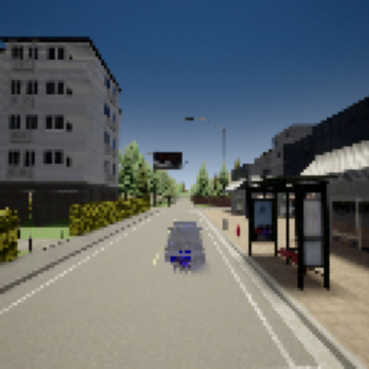}
        \label{Figure:Town4BOutput}}
    \hfill
    \subfigure[Synthesized View for \Cref{Figure:Town5Input}]
        {\includegraphics[width=0.23\textwidth]{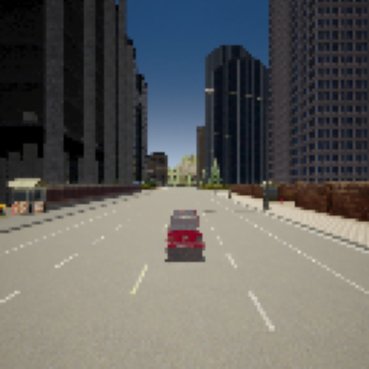}
        \label{Figure:Town5Output}}
    \caption{Synthesized Views for a Fixed $\bar{s}=10m$ and Varying Camera Views from the Training Set.}
    \label{Figure:SynthesizedViewsFixedS}
\end{figure*}

\begin{figure*}[!htb]
    \centering
    \subfigure[Camera View]
        {\includegraphics[width=0.23\textwidth]{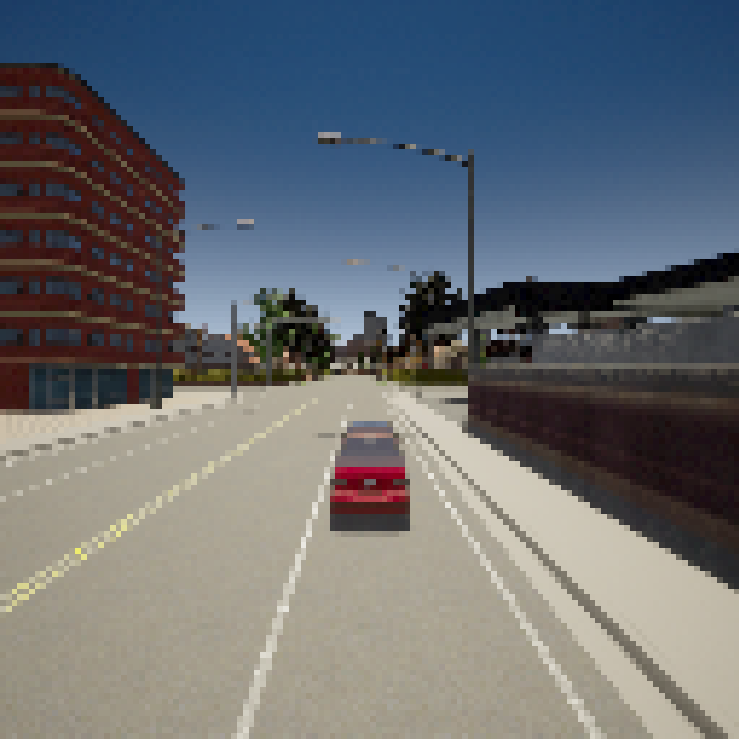}
        \label{Figure:Town3GasStationInputG}}
    \hfill
    \subfigure[Camera View]
        {\includegraphics[width=0.23\textwidth]{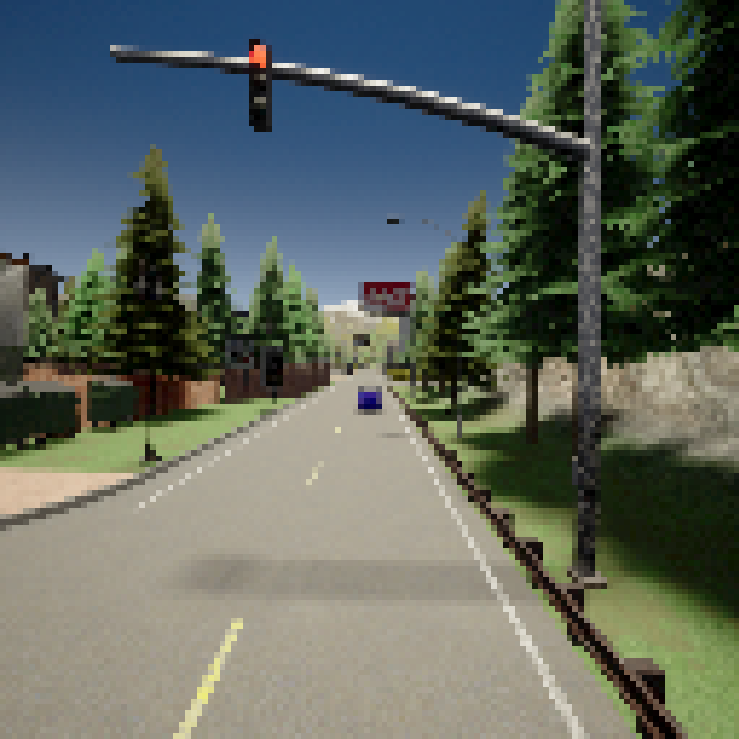}
        \label{Figure:Town4InputG}}
    \hfill
    \subfigure[Camera View]
        {\includegraphics[width=0.23\textwidth]{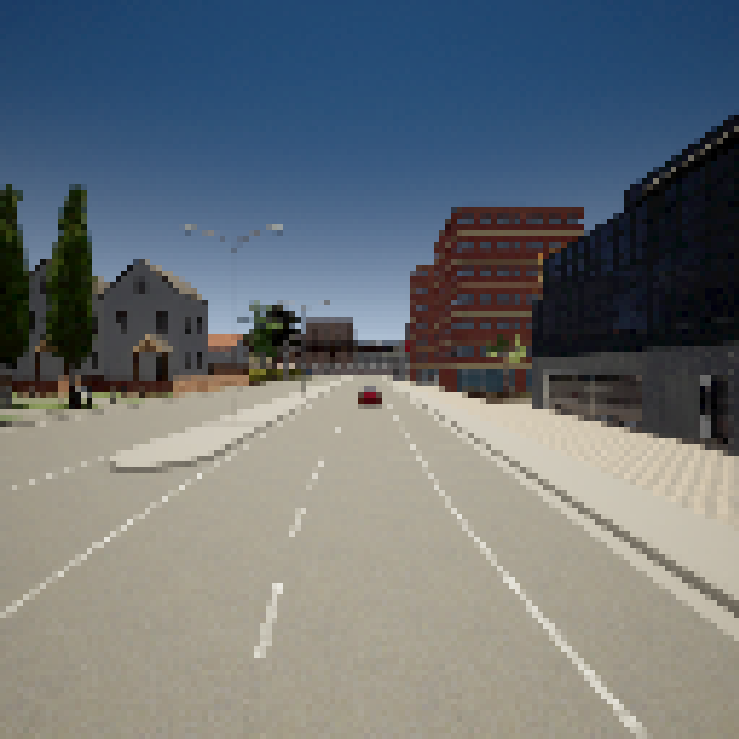}
        \label{Figure:Town3CircleInputG}}
    \hfill
    \subfigure[Camera View]
        {\includegraphics[width=0.23\textwidth]{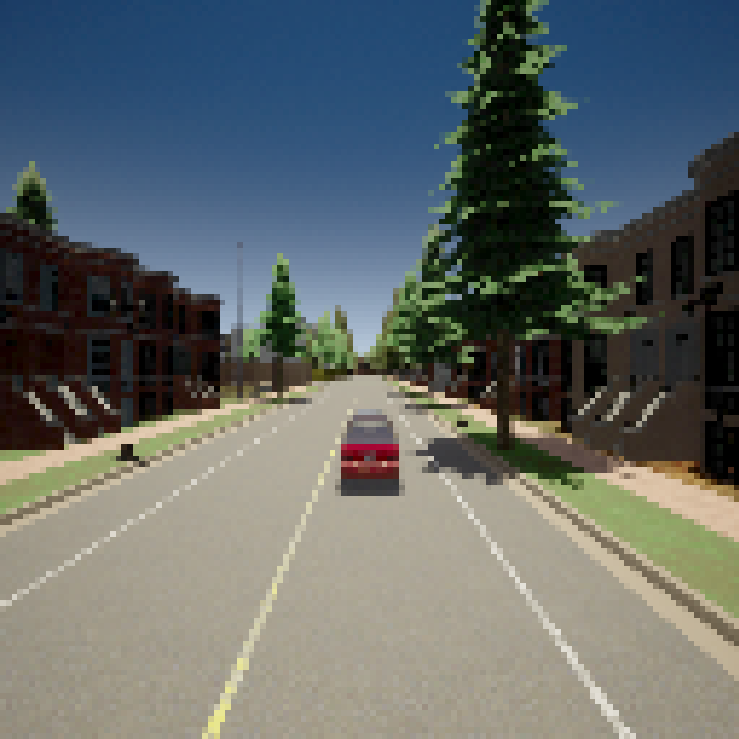}
        \label{Figure:Town5InputG}}
    \hfill
    \subfigure[Synthesized View for \Cref{Figure:Town3GasStationInputG}]
        {\includegraphics[width=0.23\textwidth]{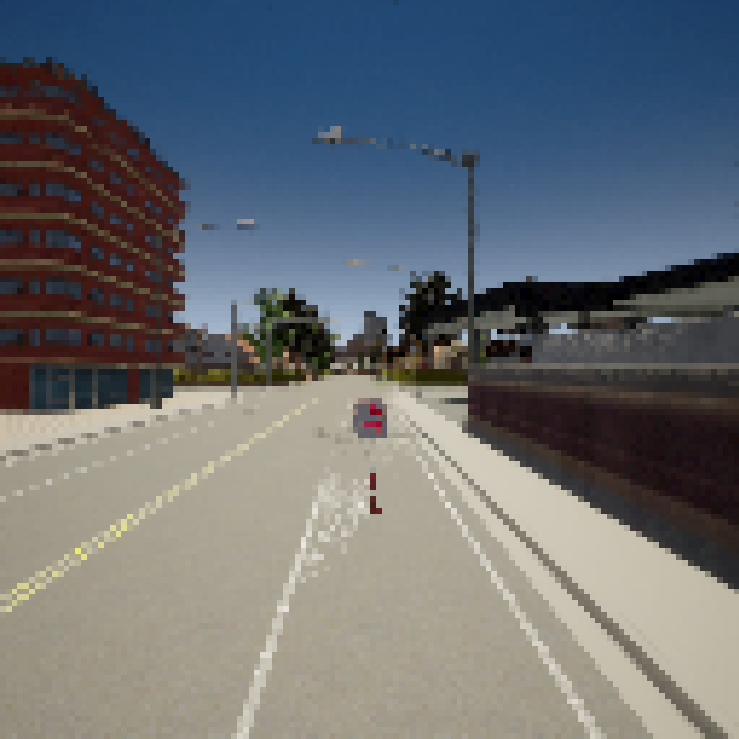}
        \label{Figure:Town3GasStationOutputG}}
    \hfill
    \subfigure[Synthesized View for \Cref{Figure:Town4InputG}]
        {\includegraphics[width=0.23\textwidth]{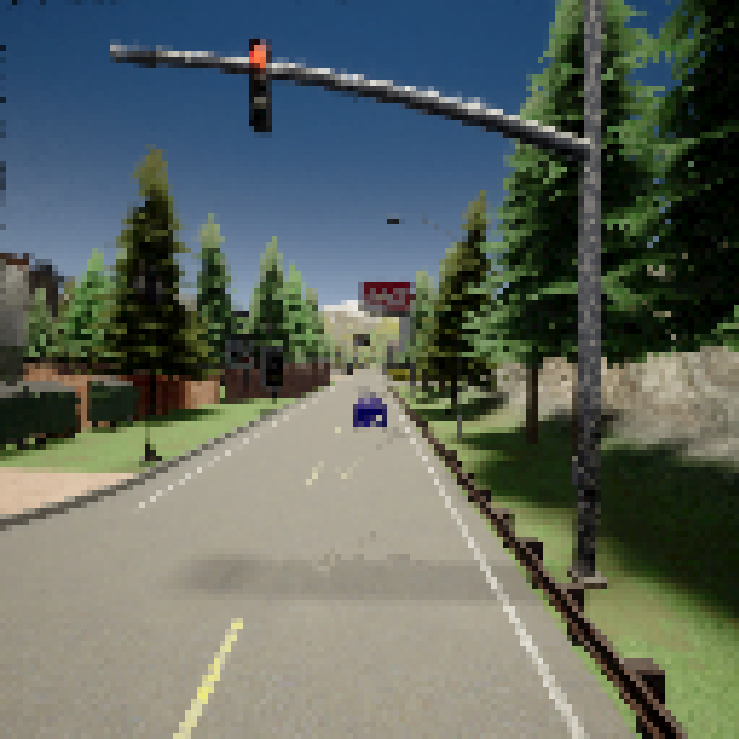}
        \label{Figure:Town4OutputG}}
    \hfill
    \subfigure[Synthesized View for \Cref{Figure:Town3CircleInputG}]
        {\includegraphics[width=0.23\textwidth]{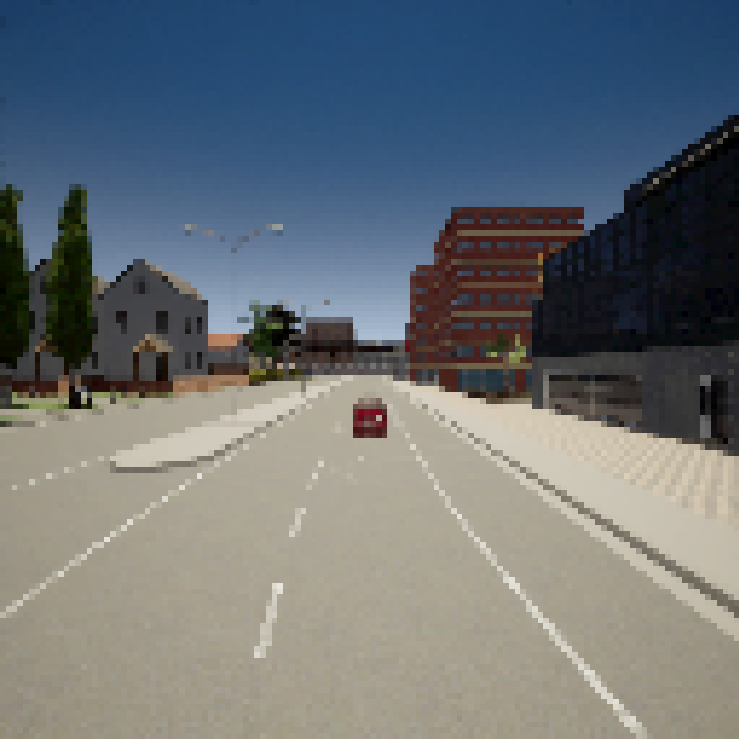}
        \label{Figure:Town3CircleOutputG}}
    \hfill
    \subfigure[Synthesized View for \Cref{Figure:Town5InputG}]
        {\includegraphics[width=0.23\textwidth]{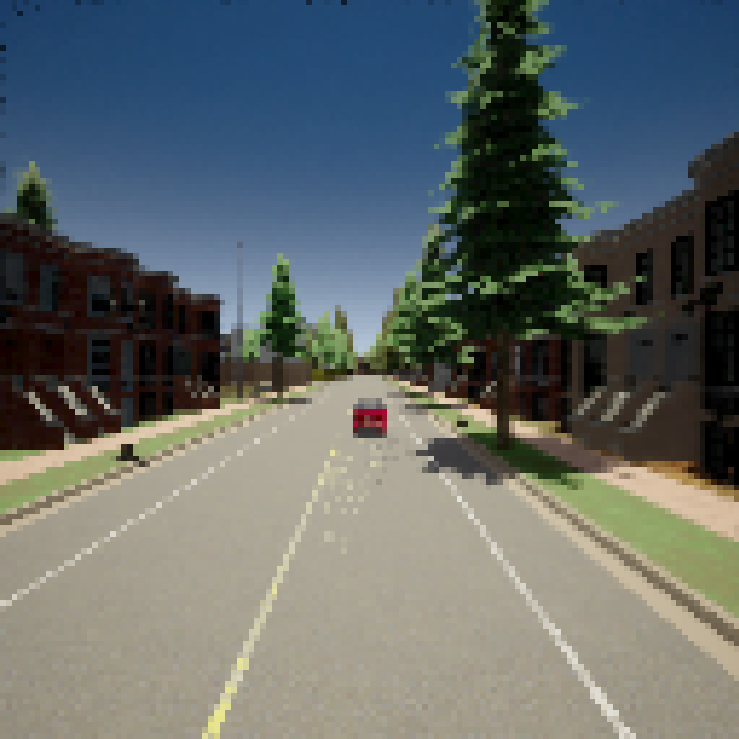}
        \label{Figure:Town5OutputG}}
    \caption{Synthesized Views for a Fixed $\bar{s}=20m$ and Varying Camera Views from Outside the Training Set.}
    \label{Figure:SynthesizedViewsFixedSGeneralization}
\end{figure*}

\begin{figure*}[!htb]
    \centering
    \subfigure[Camera View]
        {\includegraphics[width=0.23\textwidth]{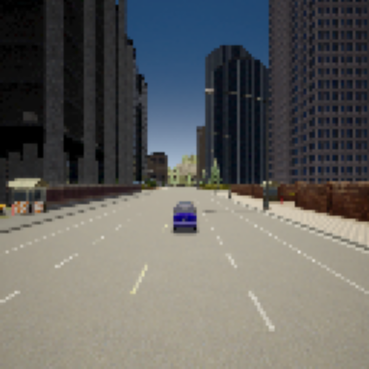}
        \label{Figure:Observed102030}}
    \hfill
    \subfigure[Synthesized View with $\bar{s}=10$]
        {\includegraphics[width=0.23\textwidth]{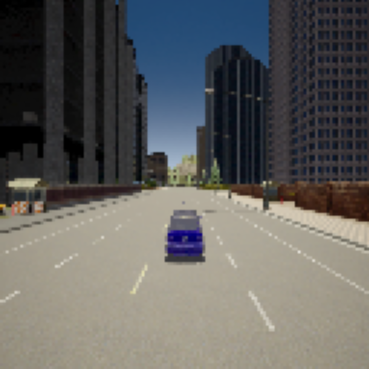}
        \label{Figure:Generated10}}
    \hfill
    \subfigure[Synthesized View with $\bar{s}=20$]
        {\includegraphics[width=0.23\textwidth]{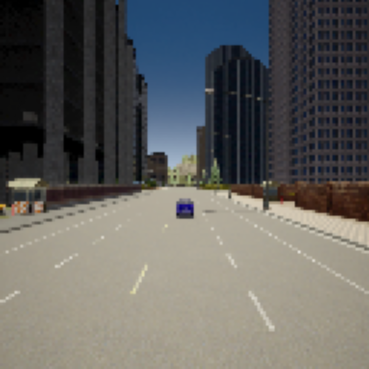}
        \label{Figure:Generated20}}
    \hfill
    \subfigure[Synthesized View with $\bar{s}=30$]
        {\includegraphics[width=0.23\textwidth]{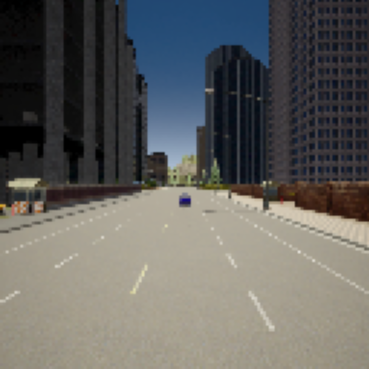}
        \label{Figure:Generated30}}
    \caption{Synthesized Views for a Fixed Camera View from the Training Set and Varying Distance $\bar{s}$.}
    \label{Figure:SynthesizedViewsVaryingS}
\end{figure*}

\begin{figure*}[!htb]
    \centering
    \subfigure[Camera View]
        {\includegraphics[width=0.23\textwidth]{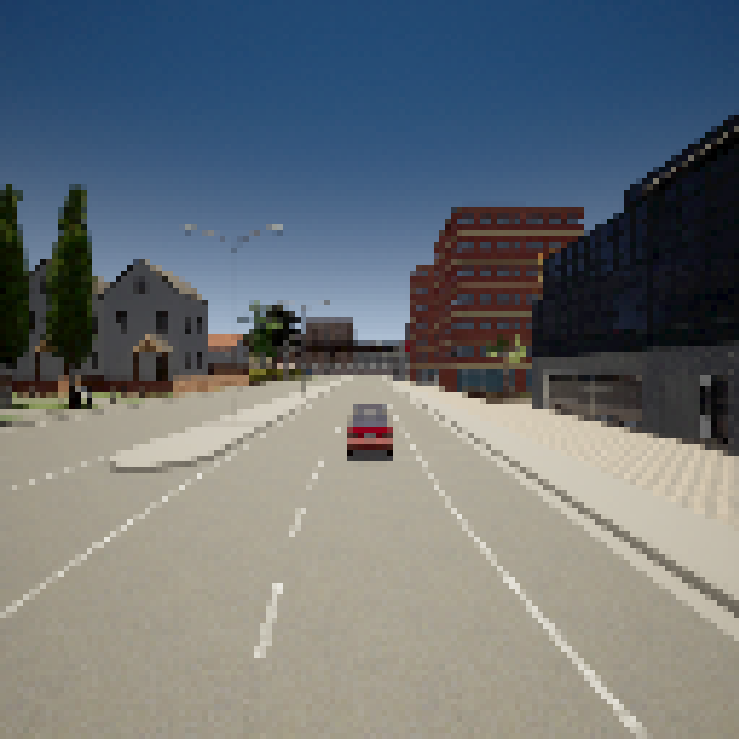}
        \label{Figure:Observed102030G1}}
    \hfill
    \subfigure[Synthesized View with $\bar{s}=10$]
        {\includegraphics[width=0.23\textwidth]{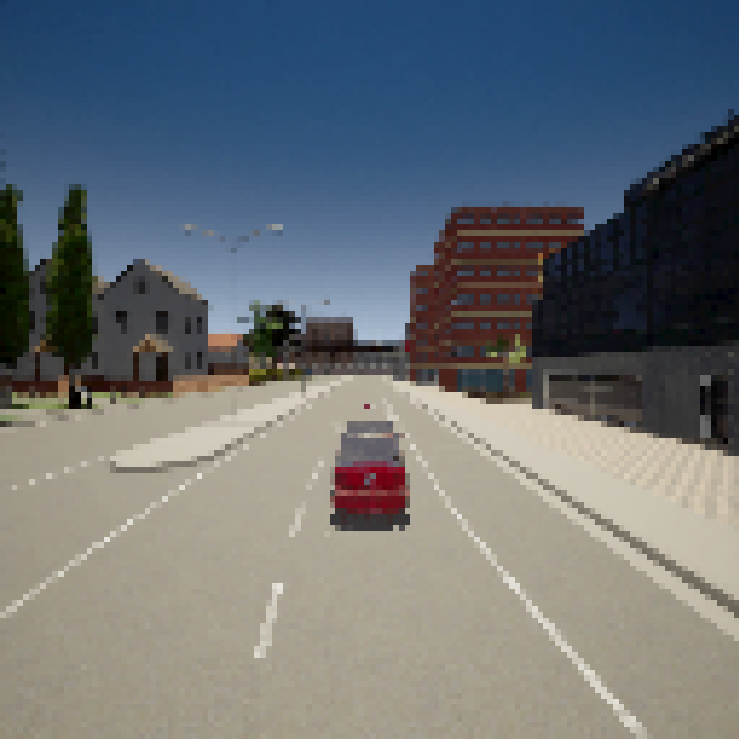}
        \label{Figure:Generated10G1}}
    \hfill
    \subfigure[Synthesized View with $\bar{s}=20$]
        {\includegraphics[width=0.23\textwidth]{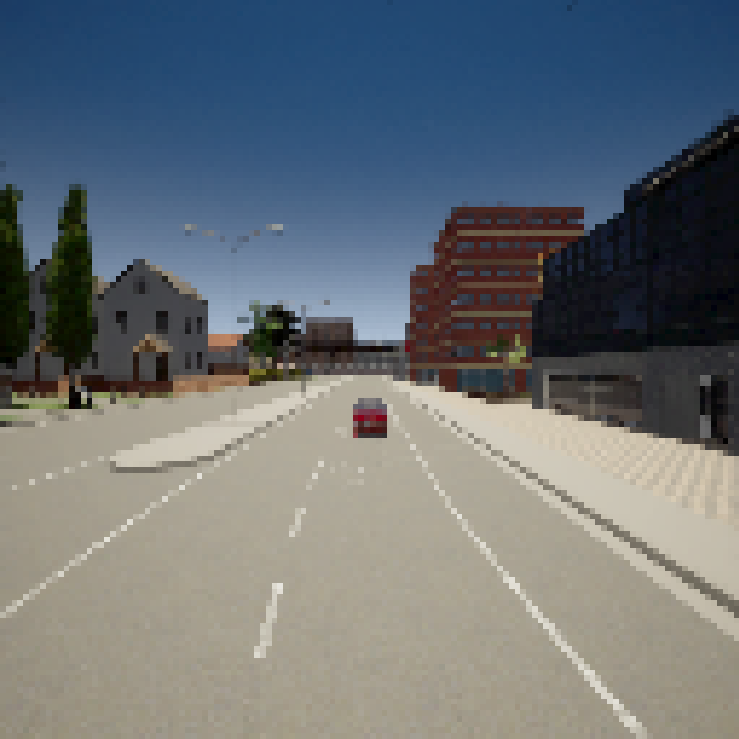}
        \label{Figure:Generated20G1}}
    \hfill
    \subfigure[Synthesized View with $\bar{s}=30$]
        {\includegraphics[width=0.23\textwidth]{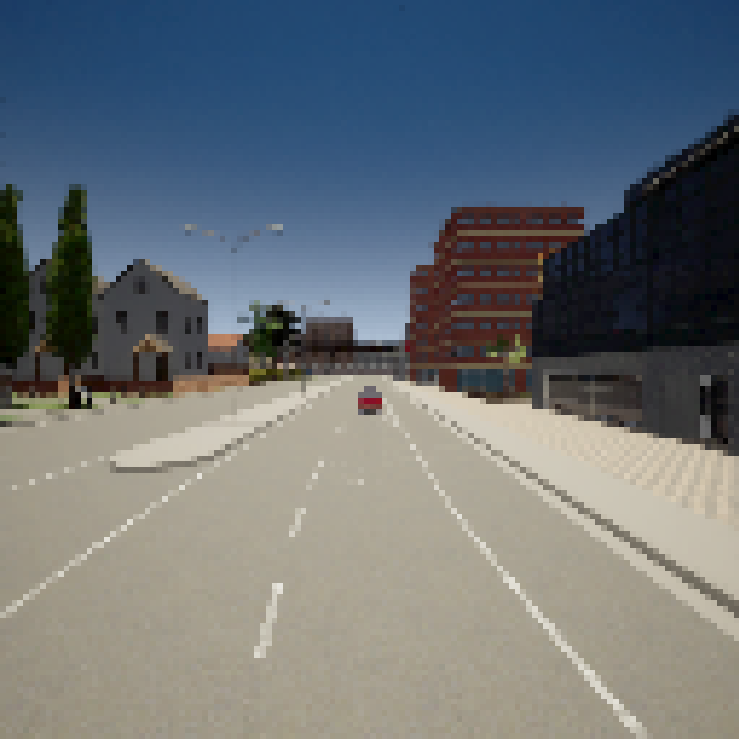}
        \label{Figure:Generated30G1}}
    \hfill
    \subfigure[Camera View]
        {\includegraphics[width=0.23\textwidth]{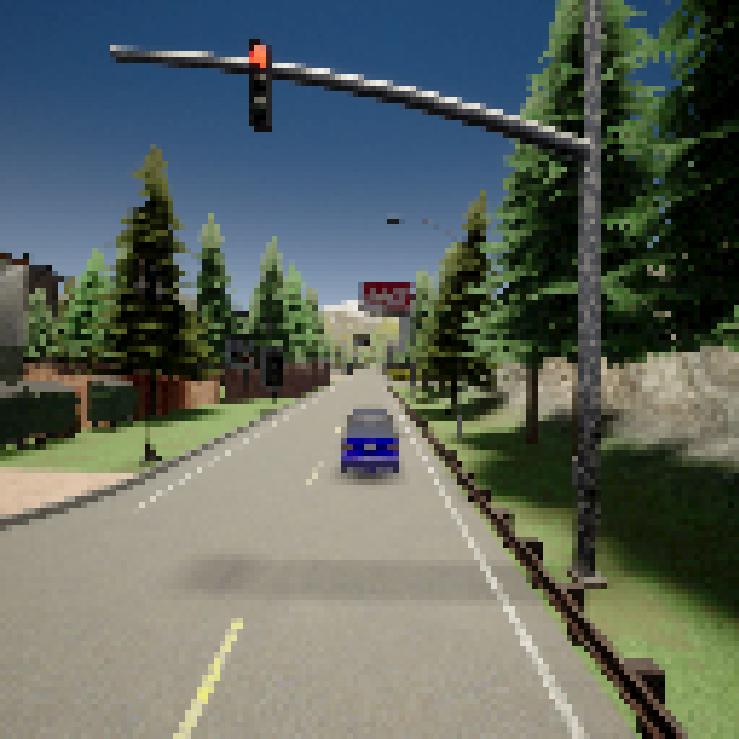}
        \label{Figure:Observed102030G2}}
    \hfill
    \subfigure[Synthesized View with $\bar{s}=10$]
        {\includegraphics[width=0.23\textwidth]{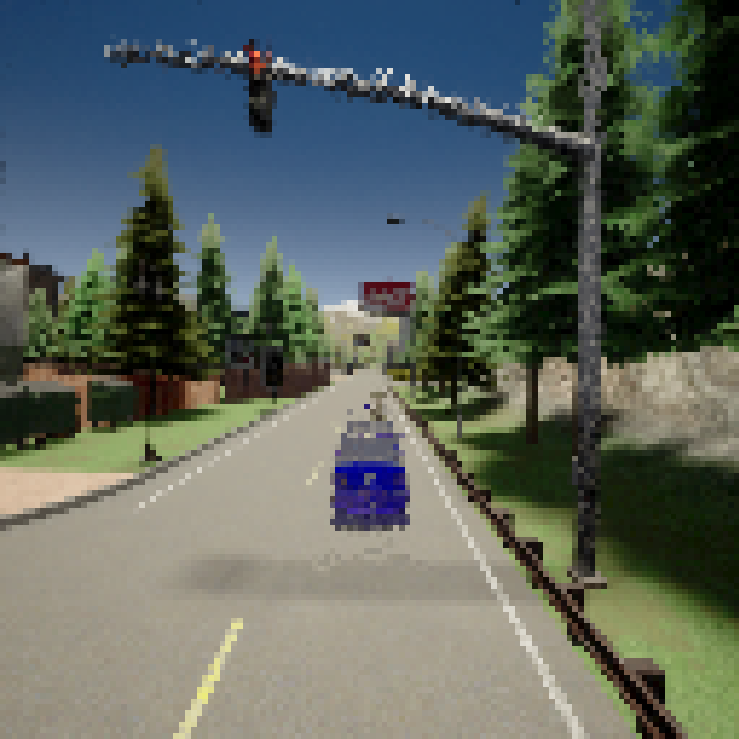}
        \label{Figure:Generated10G2}}
    \hfill
    \subfigure[Synthesized View with $\bar{s}=20$]
        {\includegraphics[width=0.23\textwidth]{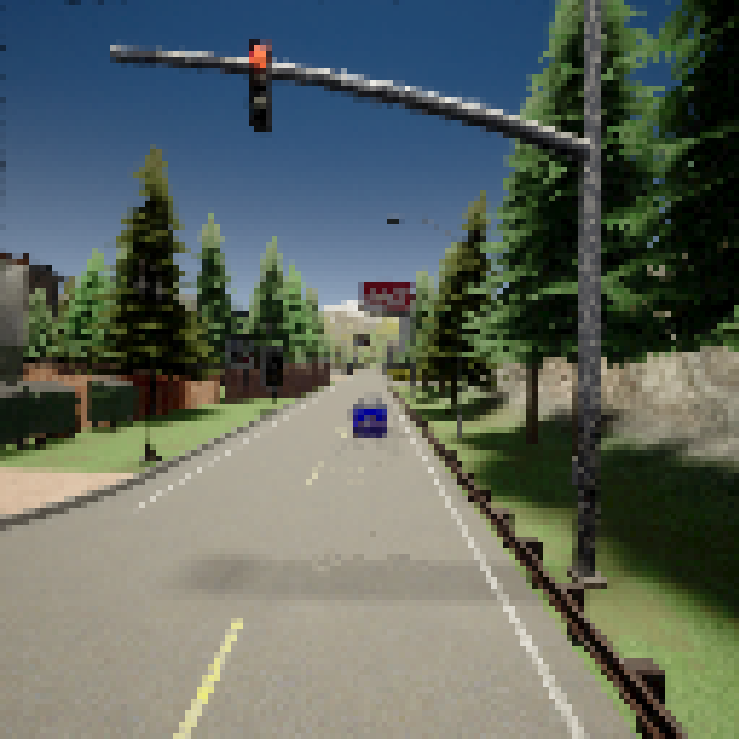}
        \label{Figure:Generated20G2}}
    \hfill
    \subfigure[Synthesized View with $\bar{s}=30$]
        {\includegraphics[width=0.23\textwidth]{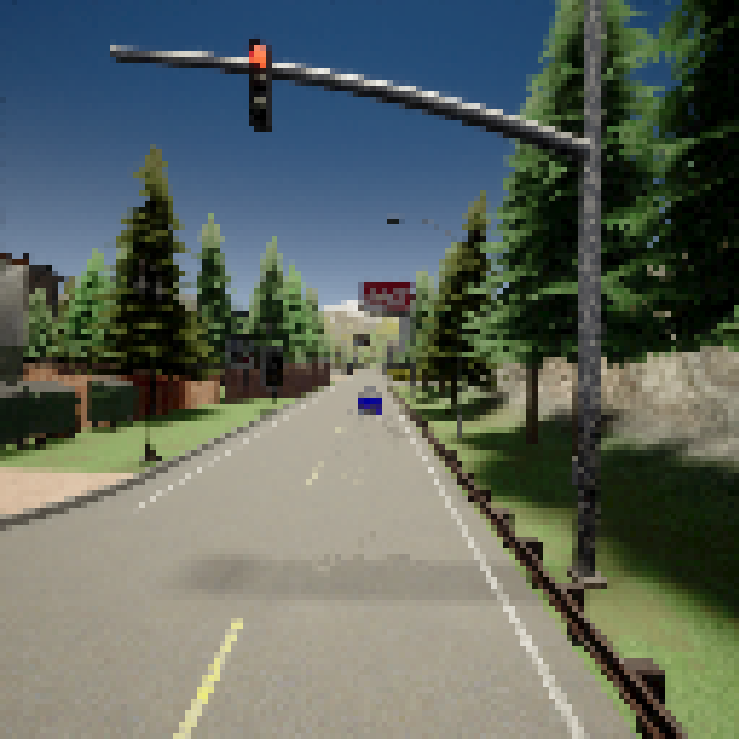}
        \label{Figure:Generated30G2}}
    \caption{Synthesized Views for Fixed Camera Views from Outside the Training Set and Varying Distance $\bar{s}$.}
    \label{Figure:SynthesizedViewsVaryingSG1}
\end{figure*}

In \Cref{Figure:SynthesizedViewsFixedS}, we show synthesized views placing the leading car at the same desired reference spacing of $\bar{s}=10$ for different camera views. The camera views are drawn from the training set. To show generalization,  \Cref{Figure:SynthesizedViewsFixedSGeneralization} shows synthesized views placing the leading car at the same desired reference spacing of $\bar{s}=20$ for different camera views. The camera views are drawn from outside the training set.

In \Cref{Figure:SynthesizedViewsVaryingS}, we show synthesized views placing the leading car at different desired reference spacings --- namely $\bar{s}=10$, $\bar{s}=20$ and $\bar{s}=30$ --- for the same camera view. The camera view is drawn from the training set. To show generalization,  \Cref{Figure:SynthesizedViewsVaryingSG1} shows synthesized views placing the leading car at different desired reference spacings --- namely $\bar{s}=10$, $\bar{s}=20$ and $\bar{s}=30$ --- for the same camera view. Each camera view is drawn from outside the training set.

\subsection{Closed-Loop Results} \label{Section:Closed-Loop}

\begin{figure*}[!htb]
    \centering
    \subfigure[Closed-loop Responses for $\bar{s}=10$]
        {\includegraphics[width=0.42\textwidth]{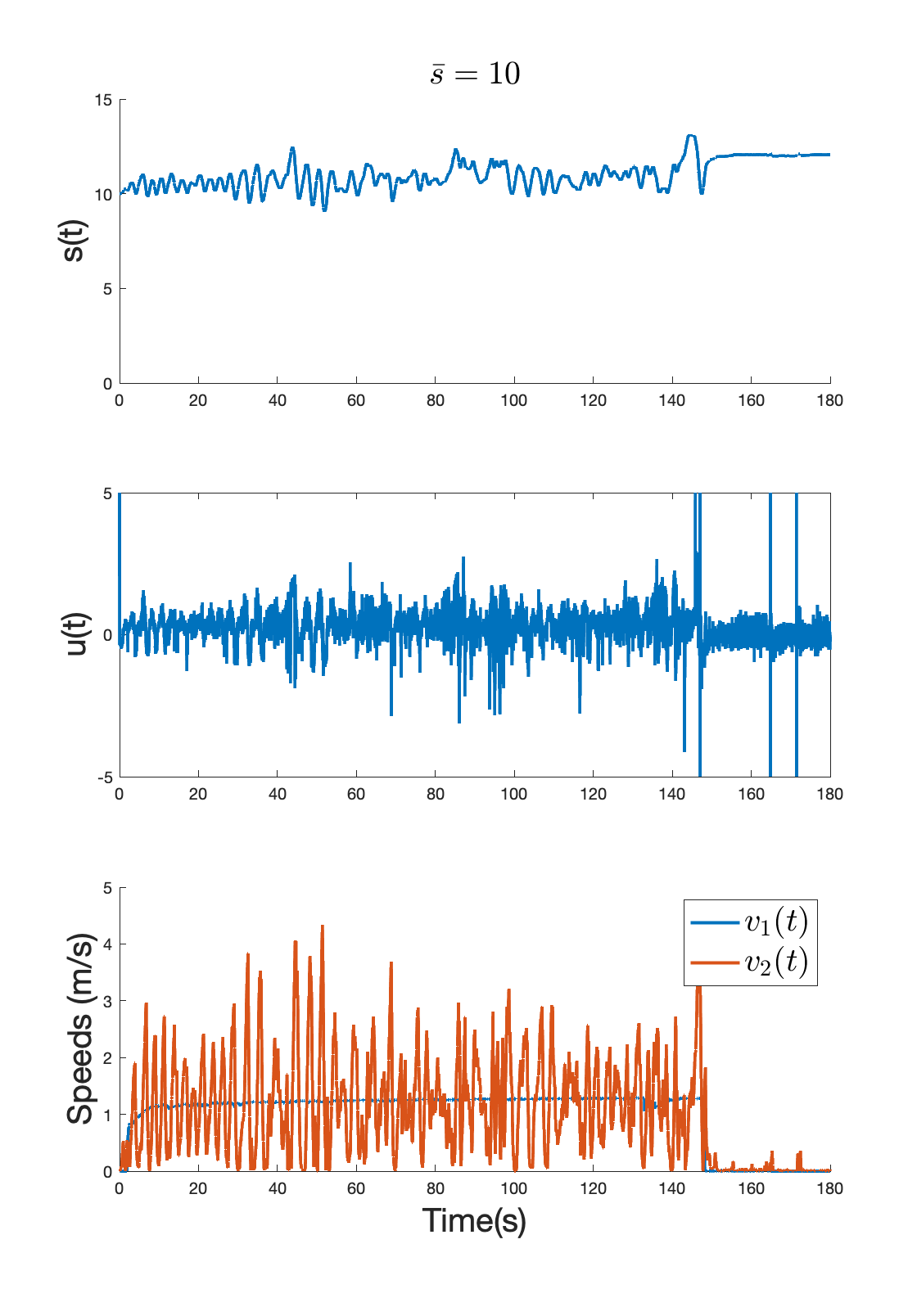}
        \label{Figure:ClosedLoopS10}}
    \hfill
    \subfigure[Closed-loop Responses for $\bar{s}=20$]
        {\includegraphics[width=0.42\textwidth]{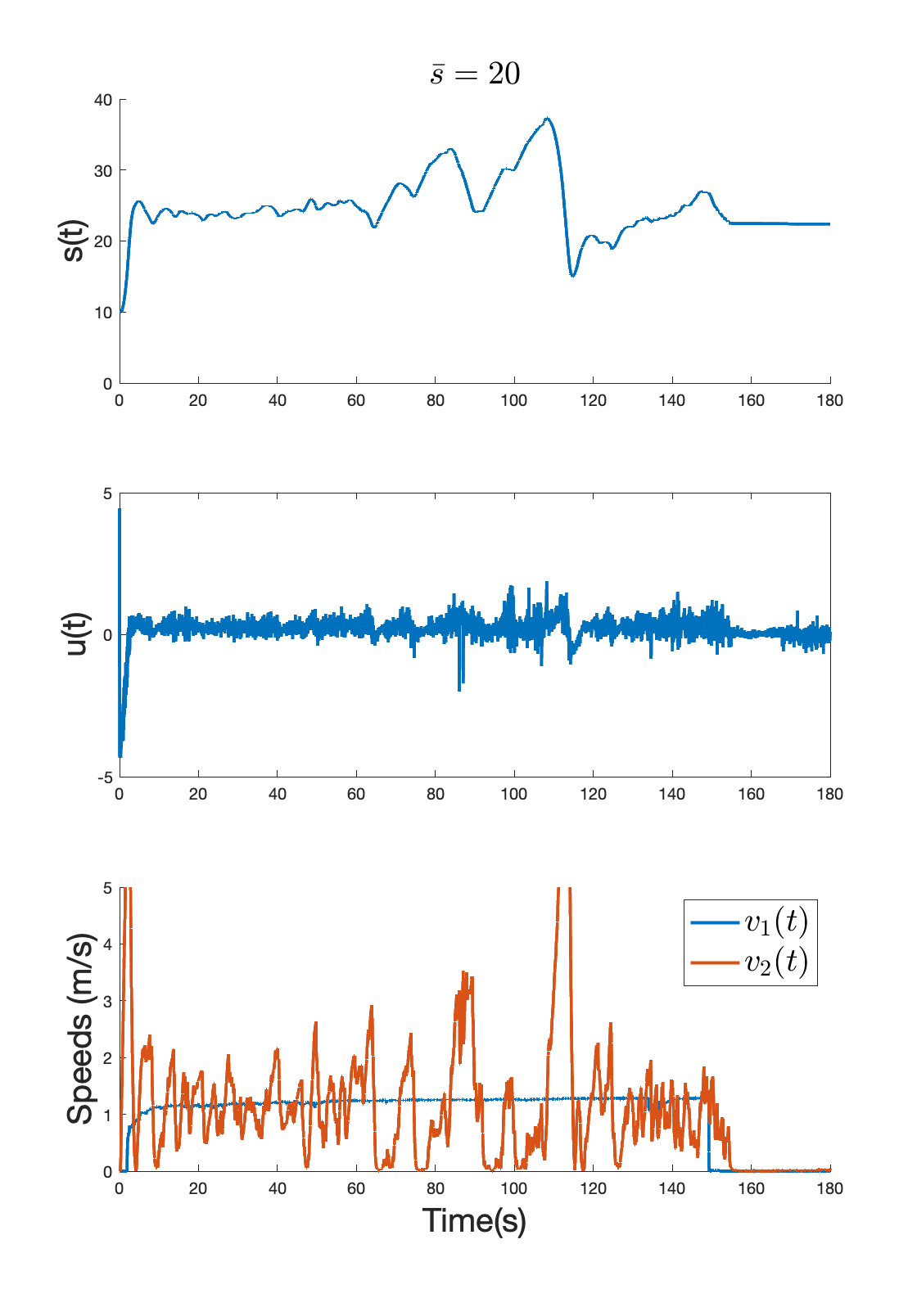}
        \label{Figure:ClosedLoopS20}}
    \caption{Closed-loop Responses.}
    \label{Figure:ClosedLoopS1020}
\end{figure*}

\Cref{Figure:ClosedLoopS1020} shows closed-loop responses for the control policy \eqref{Equation:ProportionalController} with $K = \frac{1}{50} \mathbf{1}$ using block diagram \Cref{Figure:BlockDiagramPreliminary}. The data shows that cars drive within the desired spacing $\bar{s}$. In both cases, the cars where initially spawn at an initial spacing distance of $s=10$. We have allowed the car to break and reverse direction in response to negative values of $u(t)$. Even for such a simple nonlinear controller with a proportional gain, the performance can be improved if the value of $K$ is tuned and optimized. Further, if we additionally close the loop on the speed, we can achieve tighter control.

\newpage


\bibliography{./bibtex/bib/IEEEabrv.bib,./bibtex/bib/refs.bib}

\begin{thebibliography}{35}
\providecommand{\natexlab}[1]{#1}
\providecommand{\url}[1]{\texttt{#1}}
\expandafter\ifx\csname urlstyle\endcsname\relax
  \providecommand{\doi}[1]{doi: #1}\else
  \providecommand{\doi}{doi: \begingroup \urlstyle{rm}\Url}\fi

\bibitem[{Amini} et~al.(2018){Amini}, {Schwarting}, {Rosman}, {Araki},
  {Karaman}, and {Rus}]{8594386}
A.~{Amini}, W.~{Schwarting}, G.~{Rosman}, B.~{Araki}, S.~{Karaman}, and
  D.~{Rus}.
\newblock Variational autoencoder for end-to-end control of autonomous driving
  with novelty detection and training de-biasing.
\newblock In \emph{2018 IEEE/RSJ International Conference on Intelligent Robots
  and Systems (IROS)}, pages 568--575, Oct 2018.
\newblock \doi{10.1109/IROS.2018.8594386}.

\bibitem[{Amini} et~al.(2020){Amini}, {Gilitschenski}, {Phillips}, {Moseyko},
  {Banerjee}, {Karaman}, and {Rus}]{8957584}
A.~{Amini}, I.~{Gilitschenski}, J.~{Phillips}, J.~{Moseyko}, R.~{Banerjee},
  S.~{Karaman}, and D.~{Rus}.
\newblock Learning robust control policies for end-to-end autonomous driving
  from data-driven simulation.
\newblock \emph{IEEE Robotics and Automation Letters}, 5\penalty0 (2):\penalty0
  1143--1150, 2020.
\newblock \doi{10.1109/LRA.2020.2966414}.

\bibitem[{Astolfi} and {Colaneri}(2002)]{1137557}
A.~{Astolfi} and P.~{Colaneri}.
\newblock A hamilton-jacobi setup for the static output feedback stabilization
  of nonlinear systems.
\newblock \emph{IEEE Transactions on Automatic Control}, 47\penalty0
  (12):\penalty0 2038--2041, Dec 2002.
\newblock ISSN 2334-3303.
\newblock \doi{10.1109/TAC.2002.805680}.

\bibitem[Astolfi and Colaneri(2001)]{10.1007/BFb0110207}
Alessandro Astolfi and Patrizio Colaneri.
\newblock Static output feedback stabilization: from linear to nonlinear and
  back.
\newblock In Alberto Isidori, Fran{\c{c}}oise Lamnabhi-Lagarrigue, and Witold
  Respondek, editors, \emph{Nonlinear control in the Year 2000}, pages 49--71,
  London, 2001. Springer London.
\newblock ISBN 978-1-84628-568-4.

\bibitem[Banijamali et~al.(2018)Banijamali, Shu, Ghavamzadeh, Bui, and
  Ghodsi]{pmlr-v84-banijamali18a}
Ershad Banijamali, Rui Shu, Mohammad Ghavamzadeh, Hung Bui, and Ali Ghodsi.
\newblock Robust locally-linear controllable embedding.
\newblock In Amos Storkey and Fernando Perez-Cruz, editors, \emph{Proceedings
  of the Twenty-First International Conference on Artificial Intelligence and
  Statistics}, volume~84 of \emph{Proceedings of Machine Learning Research},
  pages 1751--1759, Playa Blanca, Lanzarote, Canary Islands, 09--11 Apr 2018.
  PMLR.
\newblock URL \url{http://proceedings.mlr.press/v84/banijamali18a.html}.

\bibitem[{Bateux} et~al.(2018){Bateux}, {Marchand}, {Leitner}, {Chaumette}, and
  {Corke}]{8461068}
Q.~{Bateux}, E.~{Marchand}, J.~{Leitner}, F.~{Chaumette}, and P.~{Corke}.
\newblock Training deep neural networks for visual servoing.
\newblock In \emph{2018 IEEE International Conference on Robotics and
  Automation (ICRA)}, pages 3307--3314, 2018.
\newblock \doi{10.1109/ICRA.2018.8461068}.

\bibitem[Bojarski et~al.(2016)Bojarski, Testa, Dworakowski, Firner, Flepp,
  Goyal, Jackel, Monfort, Muller, Zhang, Zhang, Zhao, and
  Zieba]{DBLP:journals/corr/BojarskiTDFFGJM16}
Mariusz Bojarski, Davide~Del Testa, Daniel Dworakowski, Bernhard Firner, Beat
  Flepp, Prasoon Goyal, Lawrence~D. Jackel, Mathew Monfort, Urs Muller, Jiakai
  Zhang, Xin Zhang, Jake Zhao, and Karol Zieba.
\newblock End to end learning for self-driving cars.
\newblock \emph{CoRR}, abs/1604.07316, 2016.
\newblock URL \url{http://arxiv.org/abs/1604.07316}.

\bibitem[Chen et~al.(2020)Chen, Sax, Lewis, Armeni, Savarese, Zamir, Malik, and
  Pinto]{chen2020robust}
Bryan Chen, Alexander Sax, Gene Lewis, Iro Armeni, Silvio Savarese, Amir Zamir,
  Jitendra Malik, and Lerrel Pinto.
\newblock Robust policies via mid-level visual representations: An experimental
  study in manipulation and navigation.
\newblock \emph{4th Conference on Robot Learning (CoRL 2020)}, 2020.

\bibitem[{Collewet} and {Marchand}(2011)]{5733432}
C.~{Collewet} and E.~{Marchand}.
\newblock Photometric visual servoing.
\newblock \emph{IEEE Transactions on Robotics}, 27\penalty0 (4):\penalty0
  828--834, 2011.
\newblock \doi{10.1109/TRO.2011.2112593}.

\bibitem[Dean and Recht(2020)]{dean2020certainty}
Sarah Dean and Benjamin Recht.
\newblock Certainty equivalent perception-based control.
\newblock \emph{arXiv preprint arXiv:2008.12332}, 2020.

\bibitem[Dean et~al.(2020{\natexlab{a}})Dean, Matni, Recht, and
  Ye]{pmlr-v120-dean20a}
Sarah Dean, Nikolai Matni, Benjamin Recht, and Vickie Ye.
\newblock Robust guarantees for perception-based control.
\newblock In Alexandre~M. Bayen, Ali Jadbabaie, George Pappas, Pablo~A.
  Parrilo, Benjamin Recht, Claire Tomlin, and Melanie Zeilinger, editors,
  \emph{Proceedings of the 2nd Conference on Learning for Dynamics and
  Control}, volume 120 of \emph{Proceedings of Machine Learning Research},
  pages 350--360, The Cloud, 10--11 Jun 2020{\natexlab{a}}. PMLR.
\newblock URL \url{http://proceedings.mlr.press/v120/dean20a.html}.

\bibitem[Dean et~al.(2020{\natexlab{b}})Dean, Taylor, Cosner, Recht, and
  Ames]{dean2020guaranteeing}
Sarah Dean, Andrew~J Taylor, Ryan~K Cosner, Benjamin Recht, and Aaron~D Ames.
\newblock Guaranteeing safety of learned perception modules via
  measurement-robust control barrier functions.
\newblock \emph{arXiv preprint arXiv:2010.16001}, 2020{\natexlab{b}}.

\bibitem[Dosovitskiy et~al.(2017)Dosovitskiy, Ros, Codevilla, Lopez, and
  Koltun]{Dosovitskiy17}
Alexey Dosovitskiy, German Ros, Felipe Codevilla, Antonio Lopez, and Vladlen
  Koltun.
\newblock {CARLA}: {An} open urban driving simulator.
\newblock In \emph{Proceedings of the 1st Annual Conference on Robot Learning},
  pages 1--16, 2017.

\bibitem[Gadewadikar et~al.(2006)Gadewadikar, Lewis, and
  Abu-Khalaf]{doi:10.2514/1.16794}
Jyotirmay Gadewadikar, Frank~L. Lewis, and Murad Abu-Khalaf.
\newblock Necessary and sufficient conditions for {H-Infinity} static
  output-feedback control.
\newblock \emph{Journal of Guidance, Control, and Dynamics}, 29\penalty0
  (4):\penalty0 915--920, 2006.
\newblock \doi{10.2514/1.16794}.
\newblock URL \url{https://doi.org/10.2514/1.16794}.

\bibitem[Hafner et~al.(2019)Hafner, Lillicrap, Fischer, Villegas, Ha, Lee, and
  Davidson]{pmlr-v97-hafner19a}
Danijar Hafner, Timothy Lillicrap, Ian Fischer, Ruben Villegas, David Ha,
  Honglak Lee, and James Davidson.
\newblock Learning latent dynamics for planning from pixels.
\newblock In Kamalika Chaudhuri and Ruslan Salakhutdinov, editors,
  \emph{Proceedings of the 36th International Conference on Machine Learning},
  volume~97 of \emph{Proceedings of Machine Learning Research}, pages
  2555--2565, Long Beach, California, USA, 09--15 Jun 2019. PMLR.
\newblock URL \url{http://proceedings.mlr.press/v97/hafner19a.html}.

\bibitem[{Hirose} et~al.(2018){Hirose}, {Sadeghian}, {Vázquez}, {Goebel}, and
  {Savarese}]{8594031}
N.~{Hirose}, A.~{Sadeghian}, M.~{Vázquez}, P.~{Goebel}, and S.~{Savarese}.
\newblock {GONet}: A semi-supervised deep learning approach for traversability
  estimation.
\newblock In \emph{2018 IEEE/RSJ International Conference on Intelligent Robots
  and Systems (IROS)}, pages 3044--3051, Oct 2018.
\newblock \doi{10.1109/IROS.2018.8594031}.

\bibitem[{Hirose} et~al.(2019{\natexlab{a}}){Hirose}, {Sadeghian}, {Xia},
  {Martín-Martín}, and {Savarese}]{8624332}
N.~{Hirose}, A.~{Sadeghian}, F.~{Xia}, R.~{Martín-Martín}, and S.~{Savarese}.
\newblock {VUNet}: Dynamic scene view synthesis for traversability estimation
  using an rgb camera.
\newblock \emph{IEEE Robotics and Automation Letters}, 4\penalty0 (2):\penalty0
  2062--2069, April 2019{\natexlab{a}}.
\newblock ISSN 2377-3774.
\newblock \doi{10.1109/LRA.2019.2894869}.

\bibitem[{Hirose} et~al.(2019{\natexlab{b}}){Hirose}, {Xia}, {Martín-Martín},
  {Sadeghian}, and {Savarese}]{8750823}
N.~{Hirose}, F.~{Xia}, R.~{Martín-Martín}, A.~{Sadeghian}, and S.~{Savarese}.
\newblock Deep visual mpc-policy learning for navigation.
\newblock \emph{IEEE Robotics and Automation Letters}, 4\penalty0 (4):\penalty0
  3184--3191, Oct 2019{\natexlab{b}}.
\newblock ISSN 2377-3774.
\newblock \doi{10.1109/LRA.2019.2925731}.

\bibitem[Jaderberg et~al.(2015)Jaderberg, Simonyan, Zisserman, and
  Kavukcuoglu]{NIPS2015_33ceb07b}
Max Jaderberg, Karen Simonyan, Andrew Zisserman, and Koray Kavukcuoglu.
\newblock Spatial transformer networks.
\newblock In C.~Cortes, N.~Lawrence, D.~Lee, M.~Sugiyama, and R.~Garnett,
  editors, \emph{Advances in Neural Information Processing Systems}, volume~28,
  pages 2017--2025. Curran Associates, Inc., 2015.
\newblock URL
  \url{https://proceedings.neurips.cc/paper/2015/file/33ceb07bf4eeb3da587e268d663aba1a-Paper.pdf}.

\bibitem[Johnson et~al.(2016)Johnson, Duvenaud, Wiltschko, Adams, and
  Datta]{NIPS2016_6379}
Matthew~J Johnson, David~K Duvenaud, Alex Wiltschko, Ryan~P Adams, and
  Sandeep~R Datta.
\newblock Composing graphical models with neural networks for structured
  representations and fast inference.
\newblock In D.~D. Lee, M.~Sugiyama, U.~V. Luxburg, I.~Guyon, and R.~Garnett,
  editors, \emph{Advances in Neural Information Processing Systems 29}, pages
  2946--2954. Curran Associates, Inc., 2016.
\newblock URL
  \url{https://proceedings.neurips.cc/paper/2016/file/7d6044e95a16761171b130dcb476a43e-Paper.pdf}.

\bibitem[Khalil(2002)]{khalil2002nonlinear}
Hassan~K Khalil.
\newblock \emph{Nonlinear Systems}.
\newblock Prentice-Hall, 2002.

\bibitem[Kučera and Souza(1995)]{KUCERA19951357}
V.~Kučera and C.E.~De Souza.
\newblock A necessary and sufficient condition for output feedback
  stabilizability.
\newblock \emph{Automatica}, 31\penalty0 (9):\penalty0 1357 -- 1359, 1995.
\newblock ISSN 0005-1098.
\newblock \doi{https://doi.org/10.1016/0005-1098(95)00048-2}.
\newblock URL
  \url{http://www.sciencedirect.com/science/article/pii/0005109895000482}.

\bibitem[{Levine} and {Athans}(1966)]{1098376}
W.~{Levine} and M.~{Athans}.
\newblock On the optimal error regulation of a string of moving vehicles.
\newblock \emph{IEEE Transactions on Automatic Control}, 11\penalty0
  (3):\penalty0 355--361, 1966.
\newblock \doi{10.1109/TAC.1966.1098376}.

\bibitem[{Nagai} and {Sakai}(2013)]{6736325}
K.~{Nagai} and S.~{Sakai}.
\newblock A visual feedback design on matrix space for a liquid sloshing
  experiment.
\newblock In \emph{The SICE Annual Conference 2013}, pages 2088--2093, 2013.

\bibitem[{Park} et~al.(2017){Park}, {Yang}, {Yumer}, {Ceylan}, and
  {Berg}]{8099565}
E.~{Park}, J.~{Yang}, E.~{Yumer}, D.~{Ceylan}, and A.~C. {Berg}.
\newblock Transformation-grounded image generation network for novel 3d view
  synthesis.
\newblock In \emph{2017 IEEE Conference on Computer Vision and Pattern
  Recognition (CVPR)}, pages 702--711, 2017.
\newblock \doi{10.1109/CVPR.2017.82}.

\bibitem[{Sakai} and {Ando}(2014)]{7039720}
S.~{Sakai} and M.~{Ando}.
\newblock On the visual systems control on matrix space.
\newblock In \emph{53rd IEEE Conference on Decision and Control}, pages
  2173--2178, 2014.
\newblock \doi{10.1109/CDC.2014.7039720}.

\bibitem[{Sakai} and {Sato}(2014)]{6781580}
S.~{Sakai} and M.~{Sato}.
\newblock Visual systems control on polynomial space and its application to
  sloshing problems.
\newblock \emph{IEEE Transactions on Control Systems Technology}, 22\penalty0
  (6):\penalty0 2176--2187, 2014.
\newblock \doi{10.1109/TCST.2014.2309615}.

\bibitem[Sax et~al.(2019)Sax, Emi, Zamir, Guibas, Savarese, and
  Malik]{sax2019midlevel}
Alexander Sax, Bradley Emi, Amir~R. Zamir, Leonidas Guibas, Silvio Savarese,
  and Jitendra Malik.
\newblock Mid-level visual representations improve generalization and sample
  efficiency for learning visuomotor policies.
\newblock \emph{arXiv preprint arXiv:1812.11971}, 2019.

\bibitem[{Saxena} et~al.(2017){Saxena}, {Pandya}, {Kumar}, {Gaud}, and
  {Krishna}]{7989442}
A.~{Saxena}, H.~{Pandya}, G.~{Kumar}, A.~{Gaud}, and K.~M. {Krishna}.
\newblock Exploring convolutional networks for end-to-end visual servoing.
\newblock In \emph{2017 IEEE International Conference on Robotics and
  Automation (ICRA)}, pages 3817--3823, 2017.
\newblock \doi{10.1109/ICRA.2017.7989442}.

\bibitem[Suh and Tedrake(2020)]{suh2020surprising}
H.~J.~Terry Suh and Russ Tedrake.
\newblock The surprising effectiveness of linear models for visual foresight in
  object pile manipulation.
\newblock In \emph{The 14th International Workshop on the Algorithmic
  Foundations of Robotics}, 2020.

\bibitem[Tatarchenko et~al.(2016)Tatarchenko, Dosovitskiy, and
  Brox]{10.1007/978-3-319-46478-7_20}
Maxim Tatarchenko, Alexey Dosovitskiy, and Thomas Brox.
\newblock Multi-view 3d models from single images with a convolutional network.
\newblock In Bastian Leibe, Jiri Matas, Nicu Sebe, and Max Welling, editors,
  \emph{Computer Vision -- ECCV 2016}, pages 322--337, Cham, 2016. Springer
  International Publishing.
\newblock ISBN 978-3-319-46478-7.

\bibitem[Watter et~al.(2015)Watter, Springenberg, Boedecker, and
  Riedmiller]{NIPS2015_5964}
Manuel Watter, Jost Springenberg, Joschka Boedecker, and Martin Riedmiller.
\newblock Embed to control: A locally linear latent dynamics model for control
  from raw images.
\newblock In C.~Cortes, N.~D. Lawrence, D.~D. Lee, M.~Sugiyama, and R.~Garnett,
  editors, \emph{Advances in Neural Information Processing Systems 28}, pages
  2746--2754. Curran Associates, Inc., 2015.
\newblock URL
  \url{https://proceedings.neurips.cc/paper/2015/file/a1afc58c6ca9540d057299ec3016d726-Paper.pdf}.

\bibitem[Yildirim et~al.(2020)Yildirim, Belledonne, Freiwald, and
  Tenenbaum]{Yildirimeaax5979}
Ilker Yildirim, Mario Belledonne, Winrich Freiwald, and Josh Tenenbaum.
\newblock Efficient inverse graphics in biological face processing.
\newblock \emph{Science Advances}, 6\penalty0 (10), 2020.
\newblock \doi{10.1126/sciadv.aax5979}.
\newblock URL \url{https://advances.sciencemag.org/content/6/10/eaax5979}.

\bibitem[Zhang et~al.(2019)Zhang, Vikram, Smith, Abbeel, Johnson, and
  Levine]{pmlr-v97-zhang19m}
Marvin Zhang, Sharad Vikram, Laura Smith, Pieter Abbeel, Matthew Johnson, and
  Sergey Levine.
\newblock {SOLAR}: Deep structured representations for model-based
  reinforcement learning.
\newblock volume~97 of \emph{Proceedings of Machine Learning Research}, pages
  7444--7453, Long Beach, California, USA, 09--15 Jun 2019. PMLR.
\newblock URL \url{http://proceedings.mlr.press/v97/zhang19m.html}.

\bibitem[Zhou et~al.(2016)Zhou, Tulsiani, Sun, Malik, and
  Efros]{10.1007/978-3-319-46493-0_18}
Tinghui Zhou, Shubham Tulsiani, Weilun Sun, Jitendra Malik, and Alexei~A.
  Efros.
\newblock View synthesis by appearance flow.
\newblock In Bastian Leibe, Jiri Matas, Nicu Sebe, and Max Welling, editors,
  \emph{Computer Vision -- ECCV 2016}, pages 286--301, Cham, 2016. Springer
  International Publishing.
\newblock ISBN 978-3-319-46493-0.

\end{thebibliography}

\end{document}